\documentclass[12pt]{article}

\usepackage{longtable, ltcaption}%
\usepackage{amsmath,amsthm,amssymb}
\allowdisplaybreaks[2]
\usepackage[margin=.8in]{geometry}
\usepackage{hyperref}
\hypersetup{
    colorlinks=true,
    linkcolor=blue,
    citecolor=blue,
    urlcolor=blue
}
\usepackage{setspace,placeins,booktabs}
\usepackage{siunitx}
\usepackage[flushleft]{threeparttable}
\usepackage{cleveref, graphicx, caption, subcaption}
\usepackage[utf8]{inputenc}
\usepackage{bm}
\usepackage{comment}
\usepackage{ragged2e}
\usepackage{tikz}
\usetikzlibrary{shapes,arrows}
\usepackage{xr}
\usepackage[style=authoryear, doi=false, url=true, backend=biber, maxbibnames=100, maxcitenames=4, uniquelist=false, uniquename=false, sorting=nyt]{biblatex}
\renewbibmacro{in:}{}
\DeclareNameAlias{default}{family-given/given-family}
\newcommand{\indicator}[1]{\mathbf{1}\{#1\}}
\newcommand\independent{\protect\mathpalette{\protect\independenT}{\perp}}
    \def\independenT#1#2{\mathrel{\setbox0\hbox{$#1#2$}%
    \copy0\kern-\wd0\mkern4mu\box0}}
\newcommand{\E}{\mathbb{E}}

\renewcommand{\L}{\textrm{L}}
\renewcommand{\P}{\text{P}}
\newcommand{\F}{\textrm{F}}
\renewcommand{\d}{\text{d}}
\newcommand{\ATT}{\text{ATT}}
\newcommand{\T}{\text{T}}

\newcommand{\point}[1]{}
\newcommand{\midbar}{\,\,\middle|\,\,}

\definecolor{uabgreen}{HTML}{007A33}

\newtheorem{theorem}{Theorem}

\newtheorem{lemma}{Lemma}

\newtheorem{assumption}{Assumption}

\newtheorem{proposition}{Proposition}
\newtheorem{condition}{Condition}

\newtheorem{inneruassumption}{Assumption}
\newenvironment{namedassumption}[1]
  {\renewcommand\theinneruassumption{#1}\inneruassumption}
  {\endinneruassumption}

\crefname{figure}{Figure}{Figures}
\crefname{assumption}{Assumption}{Assumptions}
\crefname{inneruassumption}{Assumption}{Assumptions}
\crefname{condition}{Condition}{Conditions}
\crefname{appendix}{Appendix}{Appendices}

\theoremstyle{definition}
\newtheorem{definition}{Definition}

\newtheorem{remark}{Remark}

\title{\textbf{Difference-in-differences with ``bad controls''}\footnote{Some of the results in this paper were originally in ``Difference-in-differences with time-varying covariates'' \parencite{caetano-callaway-payne-santanna-2022}.  This paper and our companion paper  ``Difference-in-differences when parallel trends holds conditional on covariates'' \parencite{caetano-callaway-2025} replace that paper.  Our Supplementary Appendix is available at \url{https://bcallaway11.github.io/files/badcontrols/CCPS_2026_SA_v1.pdf}.  An \texttt{R} package, \texttt{badcontrols} \parencite{badcontrols-2026}, implementing all proposed estimators is available at \url{https://github.com/hugosantanna/badcontrols}.}}

\author{Carolina Caetano\footnote{John Munro Godfrey, Sr.~Department of Economics, University of Georgia.  \href{mailto:carol.caetano@uga.edu}{carol.caetano@uga.edu}} \quad Brantly Callaway\footnote{John Munro Godfrey, Sr.~Department of Economics, University of Georgia.  \href{mailto:brantly.callaway@uga.edu}{brantly.callaway@uga.edu}} \quad Stroud Payne\footnote{Department of Economics, Vanderbilt University. \href{mailto:robert.s.payne@vanderbilt.edu}{robert.s.payne@vanderbilt.edu}} \quad Hugo Sant'Anna\footnote{Department of Marketing, Distribution, and Economics, University of Alabama at Birmingham. \href{mailto:hsantanna@uab.edu}{hsantanna@uab.edu}}}

\begin{document}

\maketitle

\abstract{\noindent This paper considers difference-in-differences identification strategies when the parallel trends assumption holds after conditioning on covariates that may themselves be affected by the treatment (often referred to as ``bad controls''). We show that common approaches such as simply dropping bad controls are often ill-advised and develop two alternative approaches that allow bad controls to function as genuine controls despite being affected by treatment. First, we derive explicit conditions that rationalize conditioning only on pre-treatment values of the bad control, leading naturally to the \textcite{callaway-santanna-2021} estimator with pre-treatment values as covariates. Second, under a covariate unconfoundedness condition, we develop imputation and double/debiased machine learning estimators that recover the average treatment effect on the treated.  We extend these results to staggered treatment adoption, provide pre-tests for the identifying assumptions, and apply the methods to study the effects of job displacement on earnings.}

\bigskip

\bigskip


\bigskip


\vspace{100pt}

\clearpage

\normalsize

\onehalfspacing

\section{Introduction}

In this paper, we study difference-in-differences identification strategies where (i) the parallel trends assumption holds after conditioning on covariates, (ii) some or all of these covariates vary over time, and (iii) some of the time-varying covariates could themselves be affected by the treatment.  Covariates that could have been affected by participating in the treatment are often referred to as ``post-treatment'' or as ``bad controls.''  Classic examples of bad controls in labor economics include job-related covariates such as an individual's occupation, industry, or union status.  Whether to control for these variables involves a real tension.  On the one hand, the parallel trends assumption may be more plausible when it includes job-related characteristics, effectively comparing outcome paths for treated and untreated individuals in similar jobs.  On the other hand, if the treatment affects, say, an individual's occupation, then controlling for observed occupation is not the same thing as controlling for what a worker's occupation would have been absent the treatment.

The received wisdom and most common practice in empirical work is to drop bad controls from the analysis.\footnote{It is fairly common to include bad controls as a robustness check, but most main specifications do not include bad controls.  See \textcite{topel-1991,jacobson-lalonde-sullivan-1993,stevens-1997}, among many others, for labor applications that follow this approach.} For example, \textcite{angrist-pischke-2008}
note that ``...we would do better to control only for variables that are not themselves caused by [the treatment].''\footnote{\textcite{angrist-pischke-2008} discuss bad controls in a cross-sectional setting in the context of deciding whether to control for occupation when studying causal effects of graduating from college on earnings. In that case, occupation is likely to be affected by attending college and, therefore, can make comparisons in earnings among those with the same occupation who graduated or did not graduate from college hard to interpret, even if college were randomly assigned.} We agree that including a bad control directly as an additional covariate yields biased estimates of the average treatment effect on the treated ($\ATT$). However, simply excluding a bad control altogether can also be problematic.  The main issue is that it does not respect the original identification strategy.  The conditions that rationalize excluding a bad control amount to a revised parallel trends assumption that does not include the bad control at all, effectively saying that the bad control was never needed as a control in the first place.

We provide two alternative approaches to dealing with bad controls in difference-in-differences applications.  Unlike current common empirical practices, both of our approaches allow for the covariate to be affected by the treatment \textit{and} play a genuine role as a covariate (so that dropping it from the analysis is inappropriate).  First, we provide specific conditions under which it suffices to control for the pre-treatment value of the bad control.  For example, in the labor context discussed above, these conditions rationalize conditioning on pre-treatment occupation, industry, and union status.  Operationalizing this approach is straightforward, as one can simply apply the \textcite{callaway-santanna-2021} estimator using pre-treatment values of the bad control as covariates.  Second, we generalize the previous approach allowing for the possibility that the researcher needs to control for additional variables in order to account for how the treatment affects the bad control.  This leads to a more complicated estimand for the $\ATT$ involving nested conditional expectations.  In this case, we introduce new imputation, doubly robust, and double/debiased machine learning estimators.

At a high level, both of our new approaches use the following logic.  First, we treat the bad control as if it were an outcome and use an identification strategy to recover how the treatment affects the covariate (the exact details of this step are what differentiates the two approaches that we propose). From this step, we back out the distribution of the bad control absent the treatment. Then, in a second step, we use this backed-out distribution in the parallel trends assumption, effectively conditioning on the value of the bad control that would have occurred absent the treatment.

We conclude the paper with an application about the effects of job displacement on earnings, where we treat a worker's occupation score as a bad control.  First, we show that job displacement decreases workers' occupation score on average, in line with occupation score being a bad control.  Second, we show that the approaches that we introduce lead to roughly 30\% smaller estimates of the effect of job displacement on earnings relative to traditional approaches that directly include or altogether exclude the occupation score as a covariate.

The remainder of the paper is organized as follows.  \Cref{sec:setup} formalizes the setting that we consider and highlights challenges due to bad controls. \Cref{sec:traditional-approaches} characterizes the bias of the two most common empirical approaches---directly including or dropping the bad control.  \Cref{sec:identification} develops our two proposed approaches.  \Cref{sec:staggered} extends both approaches to staggered treatment adoption.  \Cref{sec:estimation-inference} develops imputation and doubly robust estimators and discusses machine learning estimation of nuisance functions.  \Cref{sec:application} applies our approaches to study the effects of job displacement on earnings.  \Cref{sec:conclusion} concludes.  Some proofs and additional results appear in the appendix, with the remaining ones provided in the Supplementary Appendix.

\subsubsection*{Related Literature} \label{lit-review}

The causal inference literature has long recognized that conditioning on variables that are themselves affected by the treatment can lead to biased estimates of target causal effect parameters (see \textcite{rosenbaum-1984} for a covariate affected by the treatment as well as \textcite{robins-greenland-1992,acharya-blackwell-sen-2016} for mediators).  In econometrics, similar to our paper, \textcite{lechner-2008,flores-lagunes-2009} use a notion of treated and untreated potential covariates (allowing for the covariates to be affected by the treatment) and propose approaches to control for counterfactual covariates.  These papers are both in the context of cross-sectional data under the assumption of unconfoundedness, while we consider a setting with panel data and where the main assumption is parallel trends.

In subsequent work to the first version of our paper, \textcite{shahn-etal-2025,renson-etal-2023} also consider difference-in-differences with time-varying covariates affected by treatment.  In these papers, for a given period, the time-varying covariate realizes, then the treatment, then the outcome.  This is in line with the statistics literature on time-varying treatments \parencite{robins-1986,robins-hernan-brumback-2000,blackwell-glynn-2018,hernan-robins-2020},\footnote{The broader time-varying treatments literature is also related to our setting, but is different in important ways.  First, in that literature, identification is mainly based on a sequential unconfoundedness assumption where treated and untreated potential outcomes are independent of treatment status conditional on pre-treatment values of covariates and possibly pre-treatment outcomes. This is in contrast to our paper, where the main identifying assumption is parallel trends.  Another difference between the current paper and much of the time-varying treatments literature is that these papers are typically primarily interested in recovering causal effects of different treatment paths (e.g., where each unit can move into or out of the treatment in each period), but often at the cost of introducing assumptions that limit treatment effect heterogeneity.  The arguments in our paper could likely be extended in this direction, but our main results are for cases with either exactly two periods or with staggered treatment adoption, both of which are common in the econometrics literature.} but it is different from the setting that we consider, where the treatment realizes, then the time-varying covariate, then the outcome.  Our setting is arguably more relevant for most social science applications and makes dealing with the bad control more challenging. \textcite{knaus-pfleiderer-2026} develop a graphical framework based on Single World Intervention Graphs (SWIGs) to formalize conditioning strategies under conditional parallel trends, and their framework can allow for bad controls.

Our paper is also broadly related to work that has used parallel trends or related assumptions in the context of mediation analysis \parencite{deuchert-huber-schelker-2019,huber-schelker-strittmatter-2022,blackwell-glynn-hilbig-phillips-2026}.  Like a mediator, the bad control in our paper can be affected by the treatment.  Unlike a mediator, however, our bad control plays an important role as a covariate for identifying effects of the treatment on the outcome.  The mediation literature is typically interested in decomposing treatment effects into direct effects of the treatment and indirect effects due to the effect of the treatment on the mediator (see \textcite{huber-2020} for a review of this literature).  Our paper is less ambitious on this front in that we only seek to account for the bad control being affected by the treatment, while identifying the overall effect of the treatment on outcomes. It would be interesting to extend our arguments to additionally identifying direct and indirect effects of the treatment, where the indirect effects occur through the effect of the treatment on the bad control.

Finally, \textcite{zeldow-hatfield-2021} provide simulations assessing the performance of different approaches to controlling for covariates in difference-in-differences applications, including cases with bad controls.  \textcite{brown-butts-westerlund-2026} decompose treatment effects into direct and indirect effects in a setting with panel data, bad controls, and where the potential outcomes exhibit an interactive fixed effects structure.

\section{Setup} \label{sec:setup}

This section introduces the notation that we use in the paper, introduces the parallel trends assumption that we use throughout the paper, and explains the challenges that arise for operationalizing the parallel trends assumption in the presence of a bad control.

\subsubsection*{Notation}

For most of our results, we consider a setting with two time periods, which we label $t^*$ and $t^*-1$.  In the first time period, no units participate in the treatment.  In the second period, some units become treated while other units remain untreated. Let $D_i$ denote whether or not a unit participates in the treatment.  Next, let $Y_{it}$ denote an individual's outcome in period $t$, and let $Y_{it}(1)$ and $Y_{it}(0)$ denote treated and untreated potential outcomes for unit $i$ in time period $t$.  Notice that $Y_{it^*-1} = Y_{it^*-1}(0)$ and $Y_{it^*} = D_i Y_{it^*}(1) + (1-D_i)Y_{it^*}(0)$; i.e., in the first period, we observe untreated potential outcomes for all units because no units have participated in the treatment yet, while in the second period, we observe treated potential outcomes for treated units and untreated potential outcomes for untreated units.  Next, we use $X_{it}$ to denote the bad control (for expositional clarity, we consider the case with a single bad control, but it is straightforward to allow $X_{it}$ to be a vector).  And, in order to allow for it to be affected by the treatment, we use $X_{it}(1)$ and $X_{it}(0)$ to denote treated and untreated potential covariates.  Like for the outcome, the observed bad control can be related to potential bad controls by $X_{it^*-1} = X_{it^*-1}(0)$ and $X_{it^*} = D_i X_{it^*}(1) + (1-D_i)X_{it^*}(0)$.  Finally, we define $Z_i$ to be a vector of other covariates that can include exogenous time-varying covariates or pre-treatment/baseline covariates.\footnote{The discussion above implicitly imposes a no-anticipation assumption, i.e., that treatments do not start affecting outcomes or the bad control before the treatment takes place.  See \textcite{callaway-2023} for more discussion.  We also implicitly impose SUTVA, i.e., that the potential outcomes are well-defined and do not depend on the treatments of other units.  Finally, throughout the paper we assume that all expectations exist and that statements about conditional expectations should be taken to hold almost surely.}

\begin{assumption}[Random Sampling] \label{ass:sampling}
The observed data $\{Y_{it^*}, Y_{it^*-1}, X_{it^*}, X_{it^*-1}, Z_i, D_i\}_{i=1}^n$ are independent and identically distributed.
\end{assumption}

Following the vast majority of the difference-in-differences literature, we target identifying the average treatment effect on the treated ($\ATT$), given by
\begin{align*}
    \ATT := \E[Y_{t^*}(1) - Y_{t^*}(0) \mid D=1]
\end{align*}

Next, we move to providing a formal definition of a bad control.  Consider the following conditions:

\begin{condition}[Outcome Relevance] \label{cond:relevance}
    For $d \in \{0,1\}$,
    \begin{align*}
        \E[\Delta Y_{t^*}(0) \mid X_{t^*}(0), X_{t^*-1}, Z, D=d] \neq \E[\Delta Y_{t^*}(0) \mid Z, D=d].
    \end{align*}
\end{condition}

\begin{condition}[Affected by the Treatment] \label{cond:cov-affected-by-treatment}
\begin{align*}
\Big(X_{t^*}(0) \Bigm| X_{t^*-1}, Z, D=1\Big) \not \sim \Big(X_{t^*}(1) \Bigm| X_{t^*-1}, Z, D=1\Big)
\end{align*}
\end{condition}

\Cref{cond:relevance} says that $X_t$ actually affects the path of untreated potential outcomes.  This is essentially a definition of what it means for $X_t$ to be a covariate---if the path of untreated potential outcomes does not depend on $X_t$, then there is no reason to include $X_t$ in the analysis.  \Cref{cond:cov-affected-by-treatment} says that, conditional on the pretreatment covariates $(X_{t^*-1},Z)$, the distribution of the untreated potential covariate $X_{t^*}(0)$ is different from the distribution of the treated potential covariate $X_{t^*}(1)$.  In other words, it says that $X_{t^*}$ is affected by the treatment.  It is also worth comparing this condition to its complement that $\big(X_{t^*}(0) \mid X_{t^*-1}, Z, D=1\big) \sim \big(X_{t^*}(1) \mid X_{t^*-1}, Z, D=1\big)$, which \textcite{lechner-2011,caetano-callaway-2025} refer to as \textit{covariate exogeneity}.  Covariate exogeneity allows a covariate to evolve over time, but rules out any kind of systematic effect of the treatment on the covariate. A leading case where it would hold is under the condition that $X_{it^*}(1) = X_{it^*}(0)$ for all units.  \Cref{cond:cov-affected-by-treatment} relaxes covariate exogeneity, allowing the treatment to affect the covariate too.

\begin{definition}[Bad Control] \label{def:bad-control}
$X_{t}$ is a \emph{bad control} if it satisfies both \Cref{cond:relevance,cond:cov-affected-by-treatment}.
\end{definition}

\Cref{def:bad-control} provides a natural, formal definition of a bad control as a time-varying variable that is relevant for the trend in untreated potential outcomes while also being affected by the treatment.

\bigskip

Next, we introduce the parallel trends assumption, which is our main identifying assumption, as well as an overlap assumption:

\begin{assumption}[Conditional Parallel Trends] \label{ass:conditional-parallel-trends}
\begin{align*}
    \E[\Delta Y_{t^*}(0) \mid X_{t^*}(0), X_{t^*-1}, Z, D=1] = \E[\Delta Y_{t^*}(0) \mid X_{t^*}(0), X_{t^*-1}, Z, D=0].
\end{align*}
\end{assumption}

\begin{assumption}[Overlap] \label{ass:overlap} $\P(D=1 \mid X_{t^*}(0), X_{t^*-1}, Z) < 1$.
\end{assumption}

Assumption \ref{ass:conditional-parallel-trends} says that, on average, the path of untreated potential outcomes is the same for the treated group as for the untreated group after conditioning on untreated potential bad controls in both periods, $X_{t^*}(0)$ and $X_{t^*-1}$, and exogenous covariates, $Z$.  In other words, if both groups had remained untreated, then they would follow parallel trends conditional on all covariates.  This is a natural assumption that is in line with the literature when $X_t$ is not affected by the treatment \parencite{heckman-ichimura-todd-1997,abadie-2005,callaway-santanna-2021,caetano-callaway-2025} except that \Cref{ass:conditional-parallel-trends} is compatible with $X_t$ being affected by the treatment.  The untreated potential covariate $X_{t^*}(0)$ being unobserved for the treated group leads to substantive complications relative to existing identification strategies, which we discuss in detail below.  \Cref{ass:overlap} is a standard overlap condition that is common in the DiD literature (e.g., \textcite{abadie-2005}).  It implies that, for any value of the covariates (including the bad control), there exist untreated units that have those characteristics.  Without this assumption, there might not exist viable comparison units for some treated units.  Next, define\footnote{To avoid cluttered notation (though slightly abusing notation), we use $m_0(\cdot)$ to denote $\E[\Delta Y_{t^*} \mid \cdot \ , D=0]$ throughout the paper, allowing for the exact argument to change depending on the context.}
\begin{align} \label{eqn:m0}
    m_0(X_{t^*}(0),X_{t^*-1},Z) &:= \E[\Delta Y_{t^*} \mid X_{t^*}(0), X_{t^*-1}, Z, D=0],
\end{align}
which is the average change in outcomes conditional on the untreated potential bad control and exogenous covariates for the untreated group.  Since untreated potential covariates are observed for the untreated group, this term is identified.  Using standard identification arguments for identifying the $\ATT$ under \Cref{ass:conditional-parallel-trends}, it follows that
\begin{align} \label{eqn:unidentified-ATT}
    \ATT = \E[\Delta Y_{t^*}\mid D=1] - \E\left[m_0\left(X_{t^*}(0), X_{t^*-1}, Z\right) \mid D=1\right]
\end{align}
Although this is a correct expression for the $\ATT$, it does not imply that $\ATT$ is identified.  Even though the $m_0$ term is identified, it is infeasible to average it over the distribution of untreated potential covariates for the treated group.  This implies that, in general, the $\ATT$ is not identified when $X_t$ is a bad control.

To understand the arguments below, it is helpful to rewrite the preceding display in integral form:
{ \small
\begin{align} \label{eqn:att-missing-dist}
    \ATT &= \E[\Delta Y_{t^*} \mid D=1]
    - \int\int m_0\big(x_{t^*}(0),x_{t^*-1},z\big)\, \d\underbrace{\tilde{\F}_1(x_{t^*}(0) \mid x_{t^*-1}, z)} \, \d\F_1(x_{t^*-1},z)
\end{align}
}where $\tilde{\F}_1(x_{t^*}(0) \mid x_{t^*-1}, z) := \F_{X_{t^*}(0) \mid X_{t^*-1}, Z, D=1}(x_{t^*}(0) \mid x_{t^*-1}, z)$ and $\F_1(x_{t^*-1},z) := \F_{X_{t^*-1}, Z \mid D=1}(x_{t^*-1},z)$.\footnote{Similar to the notation $m_0$ above, we more generally use the notation $\F_1$ and $\F_0$ throughout the paper to denote cdfs conditional on $D=1$ and $D=0$, respectively, using the arguments to make it clear exactly what cdfs these represent.}
The problematic term in the previous display is the underlined one---because untreated potential covariates are not observed for the treated group, $\tilde{\F}_1(x_{t^*}(0) \mid x_{t^*-1}, z)$ is not identified.

An important implication of the discussion above is that, in general, even if conditional parallel trends holds, bad controls can lead to identification failure for the $\ATT$.  The approaches that we discuss below all have the common feature that they either try to recover $\tilde{\F}_1(x_{t^*}(0) \mid x_{t^*-1}, z)$ or involve layering on additional assumptions that make it irrelevant that this distribution is not known.  It is a real possibility, however, in any given application, that none of the sort of assumptions that we consider below for the bad control are plausible, leaving the $\ATT$ unidentified.

\begin{figure}[th]
    \centering
    \caption{Unrestricted Causal Graphs of DiD with Bad Controls}
    \label{fig:causal_path_bad_controls}

    \begin{subfigure}[b]{.48\textwidth}
        \centering
        \scalebox{1.25}{
        \begin{tikzpicture}[
                every node/.style={font=\normalsize},
                dashedred/.style={draw=red, dashed, line width=0.4pt},
                arrowred/.style={->, red, dashed, line width=0.4pt},
                arrow/.style={->, line width=0.4pt}
            ]

            \node (x1) at (-1,-1.75) {$X_{t^*-1}$};
            \node (z)  at ( 2.5,-1.75) {$Z$};
            \node (x2) at ( 1.71,0.342) {$X_{t^*}(0)$};
            \node (dy) at ( 4.8,0.342) {$\Delta Y_{t^*}(0)$};

            \draw[arrow] (x1) -- (x2);
            \draw[arrow] (z)  -- (x2);
            \draw[arrow] (x2) -- (dy);
            \draw[arrow] (x1) -- (dy);
            \draw[arrow] (z)  -- (dy);
        \end{tikzpicture}
        }
        \caption{SWIG with $D=0$}
    \end{subfigure}
    \hfill
        \begin{subfigure}[b]{.48\textwidth}
                    \centering
                    \scalebox{1.25}{
                    \begin{tikzpicture}[
                every node/.style={font=\normalsize},
                dashedred/.style={draw=red, dashed, line width=0.4pt},
                arrowred/.style={->, red, dashed, line width=0.4pt},
                arrow/.style={->, line width=0.4pt}
            ]

                \node (v0) at ( 1.71,0.342) {$X_{t^*}$};
                \node (v1) at (-0.777,0.297) {$D$};
                \node (v2) at ( 0.720,-1.58) {$Z$};
                \node[draw, dashed] (v3) at (-0.781,1.67) {$W$}; %
                \node (v4) at (-2.98,-1.51) {$X_{t^*-1}$};

                \draw[arrow] (v2) edge (v1);
                \draw[arrow] (v4) edge (v1);
                \draw[arrow] (v3) edge (v1);

                \draw[arrow] (v2) edge (v0);
                \draw[arrow] (v4) edge (v0);
                \draw[arrow] (v3) edge (v0);
                \draw[arrow] (v1) edge (v0);

            \end{tikzpicture}
        }
        \caption{DAG for $X_{t^*}$}
    \end{subfigure}

    \begin{justify}
    {\small \textit{Notes:}
        Panel (a) shows the SWIG for the untreated state $D=0$ and corresponds to the parallel trends assumption discussed in \Cref{ass:conditional-parallel-trends}.  Panel (b) provides the DAG for $X_{t^*}$, where the dashed box around $W$ indicates that it is an unobserved confounder.}
    \end{justify}
\end{figure}
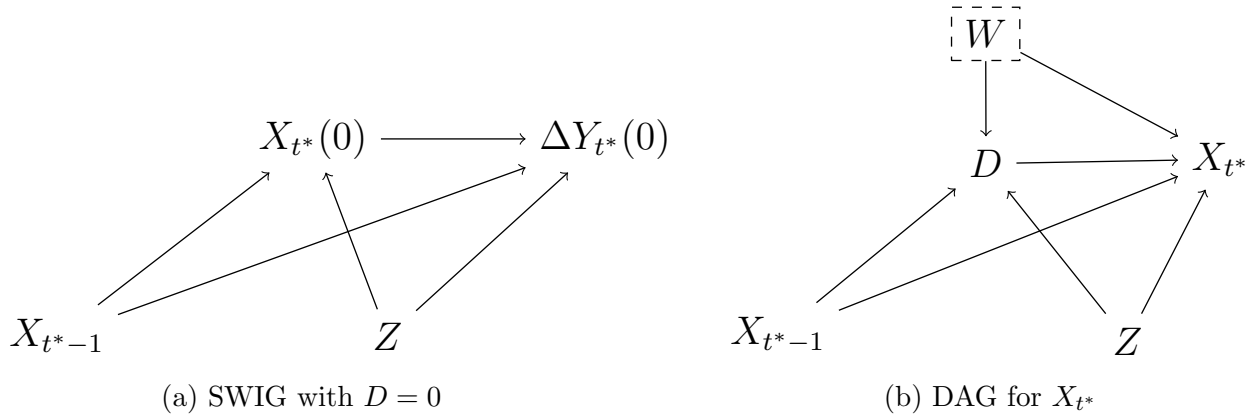

Next, we provide causal graphs describing our setting of difference-in-differences with a bad control in \Cref{fig:causal_path_bad_controls}.  Panel (a) contains a Single World Intervention Graph (SWIG, \textcite{richardson-robins-2013}).  The SWIG is essentially the graphical version of the parallel trends assumption.  It describes the counterfactual world where no units participate in the treatment.  It shows that $X_{t^*-1}$ and $Z$ affect $X_{t^*}(0)$, and that all three of $X_{t^*-1}$, $Z$, $X_{t^*}(0)$ affect $\Delta Y_{t^*}(0)$, and that this structure is common for the treated group and untreated group---if we were able to condition on $(X_{t^*}(0), X_{t^*-1}, Z)$, then the average trend in untreated potential outcomes would be the same across groups.\footnote{It is not possible to represent \Cref{ass:conditional-parallel-trends} in a standard DAG because (i) it involves untreated potential outcomes only, and (ii) it requires conditioning on $X_{t^*}(0)$ which can be different from $X_{t^*}$ in our setting.}  Panel (b) provides a Directed Acyclic Graph (DAG) for $X_{t^*}$.  In the DAG, $X_{t^*-1}$ and $Z$ are observed confounders for $X_{t^*}$, while $W$ are unobserved confounders (we use the convention of putting a dashed box around unobserved confounders).  The arrow from $D$ to $X_{t^*}$ and the presence of $W$ are what leads to $X_{t}$ being a bad control and are the sources of the identification failure of the $\ATT$ discussed above.

\FloatBarrier

\section{Traditional Approaches to Bad Controls} \label{sec:traditional-approaches}

This section examines the two traditional approaches to dealing with a bad control in empirical work---including the bad control directly, or dropping it entirely.  At a high level, we show that each approach is hard to rationalize, as each effectively requires $X_t$ not to be a bad control at all, either requiring $X_{t^*}$ not to have been affected by the treatment or not to affect the path of untreated potential outcomes.

\subsection{Approach 1: Use the Bad Control} \label{sec:approach-use-bad-control}

A first approach is to simply include the post-treatment covariate $X_{t^*}$ as if it were an ordinary control variable. This is exactly the strategy that is criticized in \textcite{angrist-pischke-2008}.  This leads to the following estimand
\begin{align*}
    \tau^{\text{use}} &:= \E[\Delta Y_{t^*}\mid D=1] - \E\left[m_0\left(X_{t^*}, X_{t^*-1}, Z\right) \mid D=1\right]\\
    &= \E[\Delta Y_{t^*} \mid D=1] - \int \int m_0(x_{t^*}, x_{t^*-1}, z) \, \d\F_1(x_{t^*} \mid x_{t^*-1}, z) \, \d\F_1(x_{t^*-1}, z)
\end{align*}
where $\F_1(x_{t^*} \mid x_{t^*-1}, z) := \F_{X_{t^*}\mid X_{t^*-1}, Z, D=1}(x_{t^*} \mid x_{t^*-1}, z)$.  $\tau^{\text{use}}$ comes from replacing the unobserved distribution of untreated potential covariates for the treated group, $\tilde{\F}_1(x_{t^*}(0) \mid x_{t^*-1}, z)$, in the expression for $\ATT$ in \Cref{eqn:unidentified-ATT} with the observed distribution of covariates for the treated group.  It is also straightforward to see the bias from using $\tau^{\text{use}}$ when $X_{t^*}$ is a bad control.  Notice that
\begin{align*}
    \tau^{\text{use}} - \ATT &= \int \int m_0(x_{t^*}(0), x_{t^*-1}, z) \, \d\Big( \tilde{\F}_1(x_{t^*}(0) \mid x_{t^*-1}, z) - \F_1(x_{t^*} \mid x_{t^*-1}, z) \Big) \, \d\F_1(x_{t^*-1}, z)
\end{align*}
which is non-zero when the distribution of $X_{t^*}(0)$ differs from the distribution of $X_{t^*}(1)$ for the treated group---this corresponds exactly to \Cref{cond:cov-affected-by-treatment} holding. Thus, this approach leads to bias whenever $X_{t^*}$ is a bad control.  \Cref{fig:dag-unaffected-covariate} shows the DAG for $X_{t^*}$ when $\tau^{\text{use}} = \ATT$.  In this and all subsequent figures that contain DAGs or SWIGs, we use red to highlight changes relative to \Cref{fig:causal_path_bad_controls}; we use dashes to indicate the removal of arrows.  Thus, the red dashed arrow from $D$ to $X_{t^*}$ indicates that the causal pathway from $D$ to $X_{t^*}$ has been removed, which, in line with the discussion above, results in $X_{t^*}$ no longer being a bad control.

\begin{figure}[th]
    \caption{DAG Rationalizing Using $X_{t^*}$ Directly as a Covariate}
    \label{fig:dag-unaffected-covariate}
    \begin{center}
    \scalebox{1.25}{
    \begin{tikzpicture}[
            every node/.style={font=\normalsize},
            dashedred/.style={draw=red, dashed, line width=0.4pt},
            arrowred/.style={->, red, dashed, line width=0.4pt},
            arrow/.style={->, line width=0.4pt}
        ]
    \node (v0) at (1.71,0.342) {$X_{t^*}$};
    \node (v1) at (-0.777,0.297) {$D$};
    \node (v2) at (0.720,-1.58) {$Z$};
    \node[draw, dashed] (v3) at (-0.781,1.67) {$W$};
    \node (v4) at (-2.98,-1.51) {$X_{t^*-1}$};
    \draw[arrow] (v2) edge (v1);
    \draw[arrow] (v4) edge (v1);
    \draw[arrowred] (v1) edge (v0);
    \draw[arrow] (v3) edge (v1);
    \draw[arrow] (v2) edge (v0);
    \draw[arrow] (v4) edge (v0);
    \draw[arrow] (v3) edge (v0);
    \end{tikzpicture}
    }
    \end{center}
    \begin{justify}
        {\small \textit{Notes:}  The figure modifies Panel (b) of \Cref{fig:causal_path_bad_controls} in a way that rationalizes the approach in this section.  The dashed red arrow from $D$ to $X_{t^*}$ indicates that this DAG removes the causal pathway from $D$ to $X_{t^*}$ that was in Panel (b) of \Cref{fig:causal_path_bad_controls}.
        }
    \end{justify}
\end{figure}
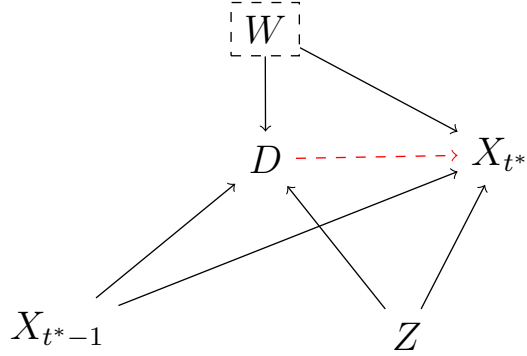

\FloatBarrier

\subsection{Approach 2: Discard Bad Controls} \label{sec:approach-discard-bad-control}

The second approach that we consider is dropping the bad control from the analysis.  This is the most common approach for dealing with bad controls in empirical work in economics.  A common rationalization for this approach is
that by avoiding $X_{t^*}$ altogether, the researcher conditions only on exogenous controls.

Dropping the bad control from the analysis leads to the following estimand:
\begin{align*}
    \tau^{\text{discard}} &:= \E[\Delta Y_{t^*}\mid D=1] - \E\left[m_0\left(Z\right) \mid D=1\right]\\
    &= \E[\Delta Y_{t^*} \mid D=1] - \int m_0(z) \d\F_1(z)
\end{align*}
It is immediately clear that this expression is different from the one for the $\ATT$ in \Cref{eqn:unidentified-ATT}.  The key difference is that, in this case, the trend in untreated potential outcomes is implicitly assumed to only depend on exogenous covariates $Z$, not on the bad control.  Figure~\ref{fig:swig_no_bad_controls} shows the implied SWIG.  In line with the preceding discussion, the SWIG drops the arrows from $X_{t^*-1}$ and $X_{t^*}(0)$ to $\Delta Y_{t^*}(0)$.

\begin{figure}[h!]
    \caption{SWIG Rationalizing Discarding $X_t$}
    \label{fig:swig_no_bad_controls}
    \begin{center}
    \scalebox{1.25}{
        \begin{tikzpicture}[
                every node/.style={font=\normalsize},
                dashedred/.style={->, red, dashed, line width=0.4pt},
                arrowred/.style={->, red, dashed, line width=0.4pt},
                arrow/.style={->, line width=0.4pt}
            ]

    \node (x1) at (-1,-1.75) {$X_{t^*-1}$};
    \node (z) at (2.5,-1.75) {$Z$};
    \node (x2) at (1.71,0.342) {$X_{t^*}(0)$};
    \node (dy) at (4.8,0.342) {$\Delta Y_{t^*}(0)$};

    \draw[arrow] (x1) -- (x2);
    \draw[dashedred] (x2) -- (dy);
    \draw[dashedred] (x1) -- (dy);
    \draw[arrow] (z) -- (x2);
    \draw[arrow] (z) -- (dy);
    \end{tikzpicture}
    }
    \end{center}
    \begin{justify}
        {\small \textit{Notes:}
            The figure modifies the SWIG in Panel (a) of \Cref{fig:causal_path_bad_controls} to rationalize dropping the bad control from the analysis altogether.  The red dashed arrows from $X_{t^*}(0)$ and $X_{t^*-1}$ to $\Delta Y_{t^*}(0)$ indicate that the causal pathway from each of these to the change in untreated potential outcomes is removed relative to Panel (a) of \Cref{fig:causal_path_bad_controls}.
        }
    \end{justify}
\end{figure}
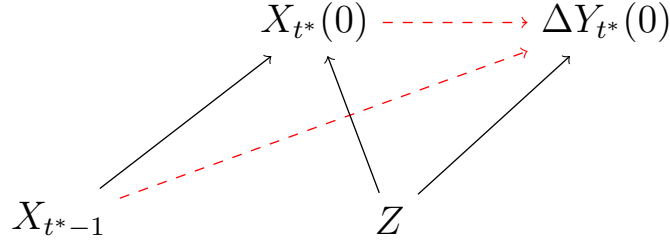

It is also straightforward to derive the bias of $\tau^{\text{discard}}$ under \Cref{ass:sampling,ass:conditional-parallel-trends,ass:overlap}.  In particular,
\begin{align*}
    \tau^{\text{discard}} - \ATT &=  \int \int \Big( m_0(x_{t^*}(0), x_{t^*-1}, z) - m_0(z) \Big) \, \d\tilde{\F}_1(x_{t^*}(0) \mid x_{t^*-1}, z) \, \d\F_1(x_{t^*-1}, z)
\end{align*}
The bias is equal to 0 when $m_0(Z) = m_0(X_{t^*}(0),X_{t^*-1},Z)$, which arises only in the case where $X_t$ does not affect the change in untreated potential outcomes. However, notice that this violates the definition of a bad control---$X_t$ is not needed as a covariate at all.  In other words, whenever $X_t$ is a bad control, dropping it from the analysis does not eliminate bias.  It only changes the form of bias relative to including the bad control.

\begin{remark}[Implications for Dealing with Bad Controls in Different Settings]
    The discussion above effectively applies to cross-sectional settings with little modification, except that there would be one period $t^*$ and by replacing $\Delta Y_{t^*}$ with $Y_{t^*}$.  The approaches that we discuss below rely on the researcher having access to panel data and would not be available (at all) with cross-sectional data and (to a large extent) with repeated cross-sections data.  Therefore, an important implication of our results is that a researcher's ability to deal with important covariates being affected by the treatment depends on the setting and nature of data that is available.
\end{remark}

\section{New Approaches to Bad Controls} \label{sec:identification}

In the previous section, we showed that the two most common approaches to dealing with bad controls essentially violate the definition of a bad control, suggesting that neither approach is attractive for applications with bad controls.  In this section, we discuss two alternative approaches for dealing with bad controls that respect the concept of a bad control by allowing it to be both affected by the treatment and to play an important role as a covariate.

In line with the discussion in the previous section, the key issue is recovering the distribution $\big(X_{t^*}(0) \mid X_{t^*-1}, Z, D=1 \big)$.  This requires additional assumptions.  We propose two approaches.  The first leads to a standard DiD estimand for the $\ATT$ that involves conditioning on $(X_{t^*-1}, Z)$, i.e., the pre-treatment value of the bad-control along with the other exogenous covariates.  Second, we propose a generalized (weaker) assumption that leads to a more complicated estimand.  In both cases, at a high level, our approach amounts to a two-step identification argument.  In the first step, we effectively treat $X_{t^*}$ as an outcome and use an identification strategy to recover its counterfactual distribution.  In the second step, having recovered its counterfactual distribution allows us to effectively control for it in the parallel trends assumption for $\Delta Y_{t^*}(0)$.

\subsection{New Approach 1: Condition on Pre-Treatment Value of the Bad Control} \label{sec:good-pre-only}

To start with, we introduce the main assumption that we use in this section:

\begin{assumption}[Simple Covariate Unconfoundedness] \label{ass:simple-cov-unc}
\begin{align*}
    X_{t^*}(0) \independent D \mid X_{t^*-1}, Z
\end{align*}
\end{assumption}

\Cref{ass:simple-cov-unc} is an unconfoundedness assumption but where $X_{t^*}$ is the outcome.  It says that the untreated potential covariate $X_{t^*}(0)$ is independent of treatment status conditional on $(X_{t^*-1}, Z)$.  Notice that \Cref{ass:simple-cov-unc} allows $X_t$ to be a genuine bad control.  We have not modified \Cref{ass:conditional-parallel-trends}, so it still affects the path of untreated potential outcomes.  It also does not put any restrictions on $X_{t^*}(1)$.  Therefore, the treatment can affect $X_{t^*}$ in arbitrary ways.  What it requires is that we can learn about how untreated potential covariates, $X_{t^*}(0)$, would have evolved absent the treatment by looking at the untreated group with the same pre-treatment value of the bad control and other exogenous covariates, $(X_{t^*-1},Z)$.  Under this additional assumption, the $\ATT$ is identified, as we show in the next theorem.

\begin{theorem} \label{thm:att-pretreatment-covs} Under \Cref{ass:sampling,ass:overlap,ass:conditional-parallel-trends,ass:simple-cov-unc},
    \begin{align*}
        \ATT = \E[\Delta Y_{t^*} \mid D=1] - \E\big[ m_0(X_{t^*-1},Z) \mid D=1 \big]
    \end{align*}
\end{theorem}

The proof is provided in \Cref{app:proofs}.  The expression for the $\ATT$ in \Cref{thm:att-pretreatment-covs} is a standard DiD estimand \parencite{heckman-ichimura-todd-1997,abadie-2005,callaway-santanna-2021}, but where the covariates include the pre-treatment bad control and other exogenous covariates.  An important implication of this is that existing arguments and tools from the DiD literature apply immediately based on this result, with the only differences amounting to the setting (bad control dealt with by \Cref{ass:simple-cov-unc}) and paying careful attention to which covariates end up being used.  Panel (a) of \Cref{fig:approach3_conditions} shows the DAG for $X_{t^*}$ under this condition.  The key restriction is that it rules out other confounders besides $X_{t^*-1}$ and $Z$ for $X_{t^*}$.

The expression for $\ATT$ in \Cref{thm:att-pretreatment-covs} is also interesting because, despite $X_t$ being a genuine bad control, the expression does not involve $X_{t^*}$ (the bad control in the post-treatment period) at all.  The explanation for this is that, under \Cref{ass:simple-cov-unc}, conditional on $(X_{t^*-1},Z)$, $X_{t^*}(0)$ is balanced between the treated and untreated group.  Because it is already balanced, it does not need to be explicitly controlled for, and controlling for $(X_{t^*-1},Z)$ suffices.

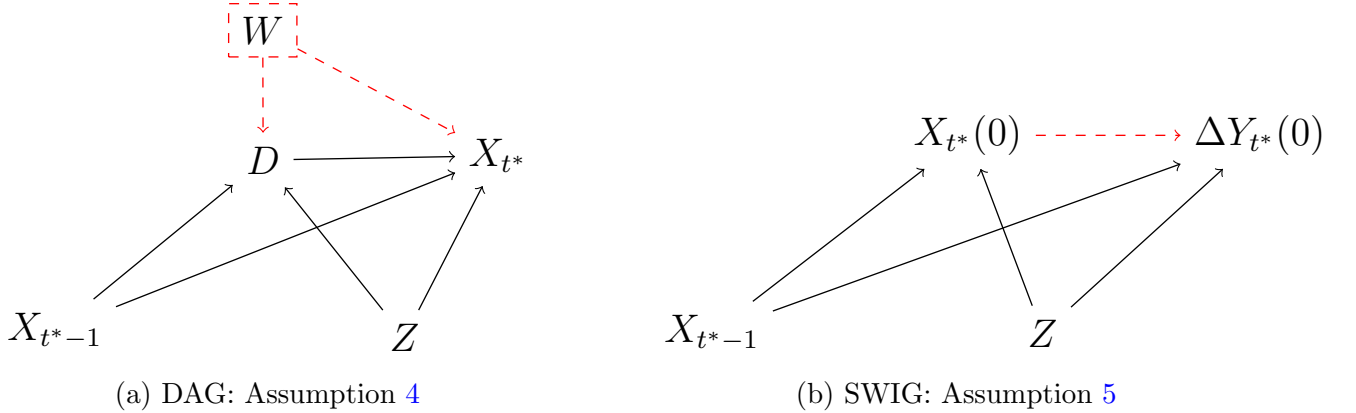
\begin{figure}[htb!]
    \centering
    \caption{Causal Graphs for Approach 3}
    \label{fig:approach3_conditions}

    \begin{subfigure}[b]{.48\textwidth}
        \centering
        \scalebox{1.25}{
        \begin{tikzpicture}[
                every node/.style={font=\normalsize},
                dashedred/.style={draw=red, dashed, line width=0.4pt},
                dashedbox/.style={draw, dashed, line width=0.4pt},
                arrowred/.style={->, red, dashed, line width=0.4pt},
                arrow/.style={->, line width=0.4pt}
            ]
            \node (v0) at ( 1.71, 0.342) {$X_{t^*}$};
            \node (v1) at (-0.777, 0.297) {$D$};
            \node (v2) at ( 0.720,-1.58 ) {$Z$};
            \node[dashedred] (v3) at (-0.781, 1.67 ) {$W$};
            \node (v4) at (-2.98, -1.51 ) {$X_{t^*-1}$};

            \draw[arrow] (v2) edge (v1);
            \draw[arrow] (v4) edge (v1);
            \draw[arrow] (v1) edge (v0);
            \draw[arrow] (v2) edge (v0);
            \draw[arrow] (v4) edge (v0);

            \draw[arrowred] (v3) edge (v1);
            \draw[arrowred] (v3) edge (v0);
        \end{tikzpicture}
        }
        \caption{DAG: \Cref{ass:simple-cov-unc}}
    \end{subfigure}
    \hfill
    \begin{subfigure}[b]{.48\textwidth}
        \centering
        \scalebox{1.25}{
        \begin{tikzpicture}[
                every node/.style={font=\normalsize},
                dashedred/.style={draw=red, dashed, line width=0.4pt},
                arrowred/.style={->, red, dashed, line width=0.4pt},
                arrow/.style={->, line width=0.4pt}
            ]

            \node (x1) at (-1,-1.75) {$X_{t^*-1}$};
            \node (z)  at ( 2.5,-1.75) {$Z$};
            \node (x2) at ( 1.71,0.342) {$X_{t^*}(0)$};
            \node (dy) at ( 4.8,0.342) {$\Delta Y_{t^*}(0)$};

            \draw[arrow] (x1) -- (x2);
            \draw[arrow] (z)  -- (x2);
            \draw[arrowred] (x2) -- (dy);
            \draw[arrow] (x1) -- (dy);
            \draw[arrow] (z)  -- (dy);
        \end{tikzpicture}
        }
        \caption{SWIG: \Cref{ass:outcome-independence}}
    \end{subfigure}
    \begin{justify}
        {\small \textit{Notes:} Panel (a) modifies the DAG in Panel (b) of \Cref{fig:causal_path_bad_controls} so that \Cref{ass:simple-cov-unc} holds.  The red dashed arrows from $W$ to $D$ and $X_{t^*}$ indicates that the causal pathway from $W$ to each of these is removed.  Panel (b) modifies the SWIG in Panel (a) of \Cref{fig:causal_path_bad_controls} so that \Cref{ass:outcome-independence} holds.  The red dashed arrow from $X_{t^*}(0)$ to $\Delta Y_{t^*}(0)$ indicates that this causal pathway is removed.}
    \end{justify}
\end{figure}

\subsubsection{Alternative Conditions for the Same Estimand}

One can also arrive at the same estimand as in \Cref{thm:att-pretreatment-covs} using alternative assumptions.  In particular, consider the following assumption:
\begin{assumption}[Bad Control Redundancy] \label{ass:outcome-independence}
\begin{align*}
\E[\Delta Y_{t^*}(0) \mid X_{t^*}(0), X_{t^*-1}, Z, D=0]
= \E[\Delta Y_{t^*}(0) \mid X_{t^*-1}, Z, D=0]
\end{align*}
\end{assumption}

\Cref{ass:outcome-independence} says that, after conditioning on the bad control in the pre-treatment period and exogenous covariates, the path of untreated potential outcomes does not depend on the bad control in post-treatment periods.  This is a kind of dimension reduction assumption that is common in settings with exogenous time-varying covariates (see \textcite{callaway-santanna-2021} and \textcite[Supplementary Appendix SA 2.7]{caetano-callaway-2025}).  In the context of job displacement, it implies that controlling for untreated potential occupation in the post-treatment period is redundant/unnecessary as long as one controls for pre-treatment occupation and other exogenous covariates.  This assumption is non-nested with \Cref{ass:simple-cov-unc}, but, like that assumption, it allows for $X_t$ to be a genuine bad control---it can be affected by the treatment in post-treatment periods and the path of untreated potential outcomes can depend on it.  Panel (b) of \Cref{fig:approach3_conditions} shows the SWIG that includes this assumption, where the key restriction is that the arrow from from $X_{t^*}(0)$ to $\Delta Y_{t^*}(0)$ is removed,  indicating that the post-treatment value of the bad control does not directly affect the $\Delta Y_{t^*}(0)$ once $(X_{t^*-1}, Z)$ are controlled for.  The following proposition shows that the $\ATT$ is identified when \Cref{ass:simple-cov-unc} is replaced by \Cref{ass:outcome-independence} and gives the same estimand as in \Cref{thm:att-pretreatment-covs}.
\begin{proposition} \label{prop:att-pretreatment-alt-assumptions}
    Under \Cref{ass:sampling,ass:overlap,ass:conditional-parallel-trends,ass:outcome-independence},
    \begin{align*}
        \ATT = \E[\Delta Y_{t^*} \mid D=1] - \E\big[ m_0(X_{t^*-1},Z) \mid D=1 \big]
    \end{align*}
\end{proposition}
The proof of \Cref{prop:att-pretreatment-alt-assumptions} is provided in \Cref{app:proofs}.

\FloatBarrier

\subsection{New Approach 2: Covariate Unconfoundedness} \label{sec:cov-unc}

In this section, we generalize the approach from the previous section.  The key restriction in the previous section was that the only confounders for $X_{t^*}$ were $X_{t^*-1}$ and $Z$, which exactly correspond to the pre-treatment covariates that show up in the parallel trends assumption.  But there is no strong \textit{ex ante} reason to suppose that the covariates should be the same for $X_{t^*}$ and $\Delta Y_{t^*}(0)$.  For example, since we take the first difference in $Y_{t^*}(0)$, the number of covariates required for parallel trends may be of lower dimension than for unconfoundedness for the bad control to hold.  Relative to the original DAG in \Cref{fig:causal_path_bad_controls}, in the previous section, we assumed that there were no additional confounders.  In this section, we instead suppose that the researcher observes the confounders $W$.  We show this graphically in \Cref{fig:dag_cov_unc}.

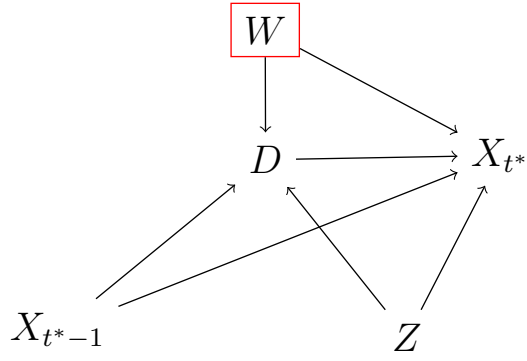
\begin{figure}[htb!]
    \centering
    \caption{DAG for $X_{t^*}$ under Covariate Unconfoundedness}
    \label{fig:dag_cov_unc}
    \scalebox{1.25}{
    \begin{tikzpicture}[
        every node/.style={font=\normalsize},
        dashedred/.style={draw=red, line width=0.4pt},
        arrowred/.style={->, red, dashed, line width=0.4pt},
        arrow/.style={->}
    ]
        \node (v0) at ( 1.71, 0.342) {$X_{t^*}$};
        \node (v1) at (-0.777, 0.297) {$D$};
        \node (v2) at ( 0.720,-1.58 ) {$Z$};
        \node[dashedred] (v3) at (-0.781, 1.67 ) {$W$};
        \node (v4) at (-2.98, -1.51 ) {$X_{t^*-1}$};

        \draw[arrow] (v2) edge (v1);
        \draw[arrow] (v4) edge (v1);
        \draw[arrow] (v3) edge (v1);

        \draw[arrow] (v2) edge (v0);
        \draw[arrow] (v4) edge (v0);
        \draw[arrow] (v3) edge (v0);

        \draw[arrow] (v1) edge (v0);
    \end{tikzpicture}
    }
    \begin{justify}
        {\small \textit{Notes:}
            The figure modifies Panel (b) of \Cref{fig:causal_path_bad_controls} so that \Cref{ass:cov-unc} holds.  The red solid box around $W$ indicates that $W$ is observed.
        }
    \end{justify}
\end{figure}
In line with the discussion above, we make the following additional assumptions:\footnote{We also assume that $W$ is observed, which expands \Cref{ass:sampling} slightly.  We do not explicitly state this expanded assumption for brevity.}
\begin{assumption}[Covariate Unconfoundedness] \label{ass:cov-unc}
    \begin{align*}
        X_{t^*}(0) \independent D \mid X_{t^*-1}, W, Z
    \end{align*}
\end{assumption}

\begin{assumption}[Overlap] \label{ass:overlap2}
$\P(D=1 \mid X_{t^*-1}, W, Z) < 1$
\end{assumption}

\Cref{ass:cov-unc} says that, conditional on the vector of pre-treatment covariates $(X_{t^*-1}, W, Z)$, the distribution of the untreated potential covariate $X_{t^*}(0)$ is the same for the treated and untreated group.  Like \Cref{ass:simple-cov-unc} (but unlike the traditional approaches discussed in \Cref{sec:traditional-approaches}), it allows for $X_{t^*}$ to both be affected by the treatment and play a genuine role as a covariate.
\Cref{ass:overlap2} modifies the overlap assumption from \Cref{ass:overlap} so that it holds conditional on $(X_{t^*-1}, W, Z)$.  Next, we show that the $\ATT$ is identified under these assumptions.
\begin{theorem} \label{thm:att-cov-unc} Under \Cref{ass:sampling,ass:conditional-parallel-trends,ass:cov-unc,ass:overlap2},
    \begin{align}
        \ATT &= \E[\Delta Y_{t^*} \mid D=1] - \E\Big[ \E\big[ m_0(X_{t^*}, X_{t^*-1}, Z) \bigm| X_{t^*-1}, W, Z, D=0 \big] \Bigm| D=1 \Big]
    \end{align}

\end{theorem}

The expression for the $\ATT$ in \Cref{thm:att-cov-unc} is more complicated than in previous results as it involves doubly nested conditional expectations, but the additional complications can be explained intuitively. The term $m_0(X_{t^*},X_{t^*-1},Z)$ is the average change in outcomes over time conditional on $X_{t^*}(0)$, $X_{t^*-1}$, and $Z$ among the untreated group.  Under \Cref{ass:conditional-parallel-trends}, this is the path of outcomes that, conditional on $X_{t^*}(0), X_{t^*-1}$, and $Z$, the treated group would have experienced if it had not participated in the treatment. The next expectation is over the distribution of $X_{t^*}(0)$ conditional on $X_{t^*-1}, W$, and $Z$ for the untreated group.  Under \Cref{ass:cov-unc}, this is the same conditional distribution that $X_{t^*}(0)$ follows for the treated group.  This part of the expression allows for $X_{t^*}$ to have been affected by the treatment, but we use observations from the untreated group with the same $(X_{t^*-1},W,Z)$ to learn how $X_{t^*}$ would have evolved for the treated group absent participating in the treatment.  Finally, the outside expectation is over the distribution of $X_{t^*-1}$, $W$, and $Z$ for the treated group and, therefore, allows for these variables to be distributed differently in the treated group than in the untreated group.

\Cref{thm:att-cov-unc} can also be explained in the context of our application on job displacement, where, for simplicity, suppose that $X_t$ denotes a worker's occupation in period $t$, $W$ denotes a worker's earnings in the pre-treatment period $t^*-1$, and $Z$ denotes a worker's age.  \Cref{ass:conditional-parallel-trends} says that, in the absence of job displacement, displaced and non-displaced workers would have experienced the same trend in earnings conditional on having the same age and the same occupation in each period.  Occupation is a bad control because it is likely to be affected by job displacement, but \Cref{ass:cov-unc} says that we can pin down what occupations displaced workers would have had in the absence of job displacement by looking at non-displaced workers who started in the same occupation and had the same age and earnings in period $t^*-1$.  \Cref{thm:att-cov-unc} effectively unpacks these assumptions to deliver the $\ATT$.  The inside expectation computes average earnings growth among non-displaced workers who match on occupation in each period and had the same age. The middle expectation then integrates over the occupations that displaced workers would have held absent displacement, recovered from non-displaced workers with the same pre-period occupation, earnings, and age. The outside expectation averages these counterfactual earnings trends over the pre-displacement characteristics of displaced workers, delivering the $\ATT$.

\begin{remark}[Lagged Outcomes as $W$]\label{rem:lagged-outcome}
An interesting special case arises when $W = Y_{t^*-1}$, so that covariate unconfoundedness holds after conditioning on the pre-treatment outcome in addition to $(X_{t^*-1}, Z)$.  In this case, \Cref{thm:att-cov-unc} yields
\begin{align*}
    \ATT &= \E[\Delta Y_{t^*} \mid D=1] - \E\Big[ \E\big[ m_0(X_{t^*}, X_{t^*-1}, Z) \big| X_{t^*-1}, Y_{t^*-1}, Z, D=0 \big] \Big| D=1 \Big].
\end{align*}
Several papers have considered the role of lagged outcomes in difference-in-differences designs (e.g., \textcite{chabe-2017,daw-hatfield-2018,imai-kim-wang-2023}).  It is common in empirical work to try to include the lagged outcome as a conditioning variable, though it can be awkward, as in many cases it either reduces the parallel trends assumption to a version of unconfoundedness conditional on the lagged outcome or leads to endogeneity similar to a dynamic panel data model (see, e.g., \textcite{marx-tamer-tang-2025}).  The expression above provides a case that is genuinely difference-in-differences, while the lagged outcome also plays an important role in dealing with bad control.
\end{remark}

\begin{remark}[Parallel Trends for Bad Control]
    For both of the approaches discussed in this section, the key additional assumptions were the unconfoundedness assumptions in \Cref{ass:simple-cov-unc} and \Cref{ass:cov-unc}, respectively.  A natural alternative would be to invoke parallel trends assumptions directly for the bad control. Importantly, however, the arguments developed above require identification of the entire conditional distribution of \(X_{t^*}(0)\) for the treated group, rather than only its conditional mean.  Methods from the difference-in-differences literature that recover the distribution of untreated potential outcomes, such as those proposed by \textcite{bonhomme-sauder-2011,callaway-li-2019,callaway-li-oka-2018}, could in principle be applied in this context, though these approaches rely on additional identifying assumptions. Similarly, change-in-changes methods \parencite{athey-imbens-2006,melly-santangelo-2015}, which are designed to recover distributions of untreated potential outcomes, could be adapted to model the evolution of time-varying covariates. A potential limitation of these distributional approaches in the present setting is that they typically point identify distributions only for continuously distributed outcomes. As a result, they may be less suitable for applications involving discrete or mixed discrete-continuous time-varying covariates, which arise frequently in empirical work.  Alternatively, under an additional linearity assumption, it is possible to recover the $\ATT$ while only recovering $\E[\Delta X_{t^*}(0) \mid D=1]$.  We explore this possibility in Supplementary Appendix \ref{app:bad-control-parallel-trends}.
\end{remark}

\section{Staggered Treatment Adoption} \label{sec:staggered}

In this section, we extend our identification results from the two-period case to settings with more time periods and variation in treatment timing across units.  We focus on the identification strategy from \Cref{sec:cov-unc}, as extending the results from \Cref{sec:good-pre-only} follows effectively immediately from using the arguments in \textcite{callaway-santanna-2021}.

\begin{namedassumption}{MP-1}[Staggered Treatment Adoption] \label{ass:staggered} For $t=2, \ldots, \T$, $D_{it-1}=1 \implies D_{it} = 1$.
\end{namedassumption}

\Cref{ass:staggered} says that, once a unit becomes treated, it remains treated in all subsequent periods.  This setting is common in empirical work and is also the most common setting that has been considered in the literature on difference-in-differences with treatment effect heterogeneity \parencite{chaisemartin-dhaultfoeuille-2020,callaway-santanna-2021,goodman-bacon-2021,sun-abraham-2021,caetano-callaway-2025}.  Let $G_i$ denote the period when unit $i$ first becomes treated, with $G_i = \infty$ for never-treated units.  We assume no group is treated in the first period (if an ``already treated'' group exists, we drop it, since parallel trends does not generally identify treatment effect parameters for this group nor is it useful as a comparison group without additional assumptions).  Let $\mathcal{G} \subseteq \{2, \ldots, \T\} \cup \{\infty\}$ denote the set of all groups and $\bar{\mathcal{G}} := \mathcal{G} \setminus \{\infty\}$ the set of groups that ever participate in the treatment. Let $Y_{it}(0)$ denote untreated potential outcomes, $Y_{it}(g)$ denote the potential outcomes in period $t$ of becoming treated in period $g$, and $X_{it}(0)$ and $X_{it}(g)$ denote the corresponding potential versions of the bad control.  Let $\mathbf{Y}_i = (Y_{i1}, \ldots,  Y_{i\T})'$ and $\mathbf{X}_i = (X_{i1}, \ldots, X_{i\T})'$
denote the $\T \times 1$ vector of outcomes across all periods and the $\T \times 1$ vector of the bad control across all time periods.  We use analogous notation for $\mathbf{D}_i$, $\mathbf{Y}_i(0)$, $\mathbf{Y}_i(g)$, $\mathbf{X}_i(0)$, and $\mathbf{X}_i(g)$, which are all $\T \times 1$ vectors.  For two periods $s_1 < s_2$, we use notation like $\mathbf{X}_{i,s_1:s_2} := (X_{s_1}, \ldots, X_{s_2})'$ to denote subvectors of $\mathbf{X}_i$ that includes only periods from $s_1$ to $s_2$.  We continue to use $Z_i$ and $W_i$ to denote additional exogenous covariates, where $W_i$ can include additional covariates that show up in the multi-period analog of \Cref{ass:cov-unc}.  Next, we introduce a no anticipation assumption.

\begin{namedassumption}{MP-2}[No anticipation]\label{ass:no-anticipation} For all units and time periods such that $t < G_i$ (i.e., pre-treatment periods for unit $i$), $Y_{it} = Y_{it}(0)$ and $X_{it} = X_{it}(0)$.
\end{namedassumption}

\Cref{ass:no-anticipation} requires that the treatment does not affect the outcome or bad control in periods before the treatment takes place. This assumption extends the standard no anticipation assumption in the DiD literature \parencite{callaway-santanna-2021,sun-abraham-2021} to additionally hold for the bad control. The assumption can be relaxed to ``limited anticipation'' at the cost of additional notation \parencite{callaway-santanna-2021}.  Under \Cref{ass:staggered,ass:no-anticipation}, observed outcomes can be expressed as $Y_{it} = \indicator{t \geq G_i} Y_{it}(G_i) + \indicator{t < G_i} Y_{it}(0)$.  In other words, in pre-treatment periods we observe untreated potential outcomes, and in post-treatment periods we observe the potential outcome corresponding to unit $i$'s actual treatment date.  Similarly, $X_{it} = \indicator{t \geq G_i} X_{it}(G_i) + \indicator{t < G_i} X_{it}(0)$.

In line with recent work that has considered difference-in-differences approaches under staggered treatment adoption, we focus on identifying group-time average treatment effects \parencite{callaway-santanna-2021,wooldridge-2025} that are defined as
\begin{align*}
    \ATT(g,t) = \E[Y_t(g) - Y_t(0) \mid G=g]
\end{align*}
$\ATT(g,t)$ is the average treatment effect for group $g$ in in period $t$.  In some applications, $\ATT(g,t)$'s can be the main target parameters, but \textcite{callaway-santanna-2021} also discuss how $\ATT(g,t)$'s can be aggregated into an overall average treatment effect, an event study, or other target parameters.  The reason to focus on identification of $ATT(g,t)$'s is that, if they can be identified, one can follow known recipes to recover all of the common target parameters in the literature.

\begin{namedassumption}{MP-3}[Multi-Period Random Sampling] \label{ass:multi-sampling} The observed data consists of $\{Y_{i1},\ldots,Y_{i\T}, X_{i1}, \ldots, X_{i\T}, W_i, Z_i, D_i\}_{i=1}^n$ which are independent and identically distributed across units.
\end{namedassumption}

\begin{namedassumption}{MP-4}[Multi-Period Conditional Parallel Trends]  \label{ass:multi-conditional-parallel-trends} For all $g \in \mathcal{G}$, and for all  $t=2,\ldots,\T$,
    \begin{align*}
        \E[\Delta Y_{t}(0) \mid \mathbf{X}(0), Z, G=g] = \E[\Delta Y_{t}(0) \mid \mathbf{X}(0), Z]
    \end{align*}
\end{namedassumption}

\begin{namedassumption}{MP-5}[Multi-Period Covariate Unconfoundedness] \label{ass:multi-cov-unc}
For all $t \geq 2$,
    \begin{align*}
        X_t(0) \independent \big(G, X_{t-2}(0), \ldots, X_1(0)\big) \bigm| \big(X_{t-1}(0), W, Z\big)
    \end{align*}
\end{namedassumption}

\begin{namedassumption}{MP-6}[Multi-Period Overlap] \label{ass:multi-overlap} For all $g \in \bar{\mathcal{G}}$, there exists some $\epsilon > 0$ such that $\P(G=g) > \epsilon$ and $\P(G=\infty \mid \mathbf{X}(0), W, Z) > \epsilon$.  %
\end{namedassumption}

\Cref{ass:multi-sampling} says that we observe $\T$ periods of panel data.  \Cref{ass:multi-conditional-parallel-trends} extends \Cref{ass:conditional-parallel-trends} to the multi-period setting.  It says that conditional parallel trends hold across all groups and time periods. \Cref{ass:multi-cov-unc} extends the two-period Covariate Unconfoundedness assumption (\Cref{ass:cov-unc}) to the staggered setting.  Conditional on the untreated potential bad control in the preceding period, $X_{t-1}(0)$, as well as the other covariates $(W,Z)$, the untreated potential bad control in the current period is independent of the group.  The assumption also embeds a first-order Markov structure, so that $X_t(0)$ depends on the entire covariate history only through $X_{t-1}(0)$.  \Cref{ass:multi-overlap} is a multiple-period version of the overlap assumption, which implies that there are available comparison groups for all values of the covariates.

\begin{theorem} \label{thm:mp} Under \Cref{ass:staggered,ass:no-anticipation,ass:multi-sampling,ass:multi-conditional-parallel-trends,ass:multi-overlap,ass:multi-cov-unc}, and for $t \geq g$ (post-treatment periods for group $g$),
    \small
    \begin{align*}
        \ATT(g,t) &= \E[Y_t - Y_{g-1} | G=g] \\
        &- \E\left[\E\big[\E[Y_t - Y_{g-1} \mid \mathbf{X}_{g:\T},\mathbf{X}_{1:(g-1)},Z,G=\infty]\bigm| \mathbf{X}_{1:(g-1)},W,Z,G=\infty \big] \Bigm| G=g\right]
    \end{align*}
\end{theorem}

\Cref{thm:mp} extends the covariate unconfoundedness result from \Cref{thm:att-cov-unc} to the staggered case, producing the same nested expectation structure.  One difference is that period $(g-1)$ is the ``base period'' rather than $t-1$.  This arises because, for group $g$, period $(g-1)$ is the period immediately before treatment begins, and parallel trends allows us to recover the path of untreated potential outcomes from that period to the current period.  \Cref{thm:mp} also uses the never-treated group as the comparison group as it is the only group for which the untreated potential bad control is observed in all periods.

\subsection{Extensions}

In this section, we provide two additional identification results for the staggered adoption setting.  One challenge for operationalizing \Cref{thm:mp} is that its identification result involves conditioning on the full vector of bad controls across periods, which will often be relatively high-dimensional.  The first extension in this section discusses (typically mild in applications) assumptions that decrease the number of periods in which the bad control needs to be conditioned on.  The second extension provides an identification result leading to conditioning only on the pre-treatment value of the bad control based on either a version of simple covariate unconfoundedness (similar to \Cref{thm:att-pretreatment-covs} above) or a redundancy condition (similar to \Cref{prop:att-pretreatment-alt-assumptions} above) but for the staggered adoption setting considered in this section.

\subsubsection{Extension 1: Dimension Reduction}

\begin{namedassumption}{MP-7}[Multi-Period Bad Control Dimension Reduction] \label{ass:mp-dim-reduction}
    For all $1 \leq t_1 < t_2 \leq \T$, and for all $g \in \mathcal{G}$,
    \begin{align*}
        \E[Y_{t_2}(0) - Y_{t_1}(0) \mid \mathbf{X}(0), Z, G=g] = \E[Y_{t_2}(0) - Y_{t_1}(0) \mid X_{t_2}(0), X_{t_1}(0), Z, G=g]
    \end{align*}
\end{namedassumption}

\Cref{ass:mp-dim-reduction} says that the path of untreated potential outcomes between any two periods depends on the time-varying covariates only through their values in those two periods, not the full history.  This type of assumption is very common in panel data models generally (see, e.g., Equation 10.12 in \textcite{wooldridge-2010}).

\begin{proposition} \label{prop:mp-lower-dimension} Under \Cref{ass:staggered,ass:no-anticipation,ass:multi-sampling,ass:multi-conditional-parallel-trends,ass:multi-cov-unc,ass:multi-overlap,ass:mp-dim-reduction}, and for $t \geq g$ (post-treatment periods for group $g$),
    \small
    \begin{align*}
        \ATT(g,t) &= \E[Y_t - Y_{g-1} | G=g] \\
        &- \E\left[\E\big[\E[Y_t - Y_{g-1} \mid X_t,X_{g-1},Z,D_t=0]\bigm| X_{g-1},W,Z,D_t=0 \big] \Bigm| G=g\right]
    \end{align*}
\end{proposition}
The key difference in \Cref{prop:mp-lower-dimension} relative to \Cref{thm:mp} is that it only involves conditioning on $X_{g-1}$ rather than the entire covariate history, with the downstream implication being that it notably simplifies estimation.  It also uses the not-yet-treated group as the comparison group, which is a larger comparison group than the never-treated group in \Cref{thm:mp}, though we note that effectively the same arguments can rationalize using the never-treated comparison group as well.  \Cref{prop:mp-lower-dimension} is a mild, but practically useful, extension of our earlier identification result, and it is the one that is implemented in our application and accompanying R package.

\subsubsection{Extension 2: Simple Covariate Unconfoundedness}

Second, we provide a staggered adoption version of the result in \Cref{thm:att-pretreatment-covs} and \Cref{prop:att-pretreatment-alt-assumptions}.

\begin{namedassumption}{MP-8}[Multi-Period Simple Covariate Unconfoundedness] \label{ass:multi-simple-cov-unc}
For all $t \geq 2$,
    \begin{align*}
        X_t(0) \independent \big(G, X_{t-2}(0), \ldots, X_1(0)\big) \bigm| \big(X_{t-1}(0), Z\big)
    \end{align*}
\end{namedassumption}

\begin{namedassumption}{MP-9}[Multi-Period Bad Control Redundancy] \label{ass:mp-bad-control-redundancy}
    For all $1 \leq t_1 < t_2 \leq \T$, and for all $g \in \mathcal{G}$,
    \begin{align*}
        \E[Y_{t_2}(0) - Y_{t_1}(0) \mid \mathbf{X}(0), Z, G=g] = \E[Y_{t_2}(0) - Y_{t_1}(0) \mid X_{t_1}(0), Z, G=g]
    \end{align*}
\end{namedassumption}

\Cref{ass:multi-simple-cov-unc} is the staggered treatment adoption version of \Cref{ass:simple-cov-unc} from the case with two periods, which provides an unconfoundedness assumption conditional on $\big(X_{g-1}(0),Z\big)$.  The difference from \Cref{ass:multi-cov-unc} is that it does not additionally condition on $W$.  \Cref{ass:mp-bad-control-redundancy} is the staggered adoption version of \Cref{ass:outcome-independence} from the two period setting.  Like that assumption, it implies that it is sufficient to control for the pre-treatment value of the bad control in the parallel trends assumption.

\begin{proposition} \label{prop:mp-simple-cov-unc} Under Assumptions \ref{ass:staggered}, \ref{ass:no-anticipation}, \ref{ass:multi-sampling}, \ref{ass:multi-conditional-parallel-trends}, \ref{ass:multi-overlap}, and \ref{ass:mp-dim-reduction}, for $t \geq g$ (post-treatment periods for group $g$), and if either Assumption \ref{ass:multi-simple-cov-unc} or \ref{ass:mp-bad-control-redundancy} holds, then
    \small
    \begin{align*}
        \ATT(g,t) &= \E[Y_t - Y_{g-1} | G=g] - \E\left[\E\big[Y_t - Y_{g-1} \bigm|  X_{g-1},Z,D_t=0 \big] \Bigm| G=g\right]
    \end{align*}
\end{proposition}

\Cref{prop:mp-simple-cov-unc} is the staggered adoption extension of \Cref{thm:att-pretreatment-covs} and \Cref{prop:att-pretreatment-alt-assumptions} from the two period case.  It says that $\ATT(g,t)$ is identified by a notably simplified estimand if either the version of unconfoundedness in \Cref{ass:multi-simple-cov-unc} holds or the redundancy condition in \Cref{ass:mp-bad-control-redundancy}.  Like the previous extension, this is basically a straightforward result given the previous ones, but it is practically useful.  Under the conditions discussed here, one can directly use \textcite{callaway-santanna-2021} for estimation as long as the pre-treatment value of the bad control is included as a covariate.

\subsection{Additional Remarks}

\begin{remark}[Aggregating $\ATT(g,t)$'s] \label{rem:aggregated-parameters}
    As discussed above, all common target parameters under staggered treatment adoption can be recovered as weighted averages of group-time average treatment effects.  For example, letting $e$ denote the length of exposure to the treatment, an event study parameter is given by
    \begin{align*}
        \ATT^{es}(e) = \sum_{g \in \mathcal{G}_e} ATT(g,g+e) \P(G=g \mid G \in \mathcal{G}_e)
    \end{align*}
    where $\mathcal{G}_e = \{ g \in \bar{\mathcal{G}} : g + e \in [2,\T]\}$, i.e., the set of groups that are observed to participate in the treatment for $e$ periods.  Similarly, an overall $\ATT$ is given by
    \begin{align*}
        \ATT^{o} = \sum_{g \in \bar{\mathcal{G}}} \sum_{t=g}^{\T} \frac{\P(G=g \mid G \in \bar{\mathcal{G}})}{\T-g+1} ATT(g,t).
    \end{align*}
    See \textcite{callaway-santanna-2021} for more details and more examples.
\end{remark}

\begin{remark}[Lagged Outcomes as $W$ in the Multi-Period Case]\label{rem:lagged-outcome-multi}
As in the two-period case (\Cref{rem:lagged-outcome}), a natural choice is $W = Y_{g-1}$, so that covariate unconfoundedness holds after conditioning on the lagged outcome.  This argument continues to go through in the staggered adoption setting considered in this section. However, an interesting (and arguably more natural) assumption in the staggered adoption setting is:
\begin{align*}
    X_t(0) \independent \big(G, X_{t-2}(0), \ldots, X_1(0)\big) \bigm| \big(X_{t-1}(0), Y_{t-1}(0), Z\big).
\end{align*}
This assumption introduces additional challenges due to $Y_{t-1}(0)$ not always being observed for group $g$, in contrast to $Y_{g-1}(0)$ always being observed.  In particular, operationalizing this assumption would require dealing with feedback from the treatment to the bad control that arises through the outcome, which is not straightforward to deal with in a difference-in-differences context.  See \textcite{bonhomme-2025,marx-tamer-tang-2025} for recent discussions about dynamics and feedback in the context of panel data causal inference.
\end{remark}

\begin{remark}[Pre-Testing]\label{rem:pseudo-att}
    None of \Cref{ass:multi-conditional-parallel-trends,ass:multi-cov-unc,ass:multi-simple-cov-unc} is directly testable, but the same sort of pre-tests that are commonly used in difference-in-differences applications are useful to assess their plausibility.  In particular, one can compute the same estimand as in Proposition \ref{prop:mp-lower-dimension} or \ref{prop:mp-simple-cov-unc} but for pre-treatment periods, in which case these pre-treatment pseudo-$\ATT(g,t)$'s should be equal to zero.  To pre-test assumptions about the bad control specifically, under our assumptions, $\ATT_X(g,t) := \E[X_t(g) - X_t(0) \mid G=g]$ is identified.  Estimating $\ATT_X(g,t)$ in pre-treatment periods provides a way to assess Assumption \ref{ass:multi-cov-unc} or \ref{ass:multi-simple-cov-unc} in the periods before the treatment occurred.  In post-treatment periods, it a test for $X_t$ actually being a bad control as it should be affected by the treatment.
\end{remark}

\section{Estimation and Inference}\label{sec:estimation-inference}

This section describes how to estimate the $\ATT$ based on the identification results in previous sections.  We focus on the case with two time periods, as the extension to estimating $\ATT(g,t)$ is straightforward.  We also focus on the setting of \Cref{thm:att-cov-unc}, where \Cref{ass:cov-unc} holds, rather than \Cref{thm:att-pretreatment-covs} because existing estimation results from \textcite{callaway-santanna-2021} hold essentially immediately in the latter case.  We propose two approaches.  First, we propose an imputation approach based on imposing linear models.  Second, we propose an approach based on a doubly robust estimand for the $\ATT$ that can be combined with estimating the nuisance functions with machine learners and, hence, permits estimation of the $\ATT$ without imposing strong functional form assumptions.

\subsection{Imputation Estimator}\label{sec:estimation-inference:ra}

We begin with an imputation estimator that serves as a natural baseline estimator that imposes linear models for unknown nuisance functions, building on the literature on imputation estimators in the context of difference-in-differences \parencite{gardner-thakral-to-yap-2023,borusyak-jaravel-spiess-2024,liu-wang-xu-2024}. This approach is easy to implement and allows for treatment effect heterogeneity.  %
From Theorem~\ref{thm:att-cov-unc}, the $\ATT$ can be written as
\begin{align*}
\ATT=\E[\Delta Y_{t^*} \mid D=1]-\tau,
\qquad
\tau :=\E\left[\nu_0(X_{t^*-1},W,Z)\mid D=1\right],
\end{align*}
where $\nu_0$ is defined as
\begin{align} \label{eqn:nu0}
    \nu_0(X_{t^*-1},W,Z) := \E[m_0(X_{t^*},X_{t^*-1},Z) \mid X_{t^*-1}, W, Z, D=0],
\end{align}
which is a nested conditional expectation in the untreated group. Our imputation approach operationalizes estimating these objects by imposing linear models for these nuisance functions.  In particular, we make the following assumption.
\begin{assumption}[Linearity Conditions for Imputation] \label{ass:imputation-linearity}
    The outcome regression for the untreated group is given by
    \begin{align*}
        m_0(X_{t^*},X_{t^*-1},Z) = X_{t^*}\beta_1 + X_{t^*-1}\beta_2+Z'\beta_3,
    \end{align*}
    and the conditional mean of the post-treatment bad control in the untreated group follows
    \begin{align*}
        \E[X_{t^*}\mid X_{t^*-1},W,Z,D=0] = X_{t^*-1}\gamma_1+W'\gamma_2 + Z'\gamma_3.
    \end{align*}
\end{assumption}
We estimate estimate $(\beta_1,\beta_2,\beta_3)$ by ordinary least squares in the untreated sample from a regression of $\Delta Y_{t^*}$ on $(X_{t^*},X_{t^*-1},Z)$, and then $(\gamma_1,\gamma_2,\gamma_3)$ by ordinary least squares in the untreated sample from a regression of $X_{t^*}$ on $(X_{t^*-1},W,Z)$. Let
\begin{align*}
    \widehat \E[X_{t^*} \mid X_{t^*-1}, W, Z, D=0]
    = X_{t^*-1}\widehat{\gamma}_1+W'\widehat\gamma_{2} + Z'\widehat\gamma_{3}
\end{align*}
denote the fitted conditional mean of $X_{t^*}$ for untreated units. Then, based on \Cref{eqn:nu0}, we construct
\begin{align*}
    \widehat\nu_0(X_{t^*-1},W,Z) = \widehat\E[X_{t^*} \mid X_{t^*-1}, W, Z, D=0] \widehat\beta_{1} + X_{t^*-1} \widehat\beta_{2} + Z'\widehat\beta_{3}.
\end{align*}

Let $n_1 :=\sum_{i=1}^n D_i$ denote the number of treated units in the sample, and define
\begin{align}
    \widehat m_1 &= \frac{1}{n_1}\sum_{i=1}^n D_i\,\Delta Y_{it^*},
    \label{eq:ra-m1hat} \\
    \widehat\tau_{ra} &= \frac{1}{n_1}\sum_{i=1}^n D_i\,\widehat\nu_0(X_{i,t^*-1},W_i,Z_i). \label{eq:ra-nu-hat}
\end{align}
The imputation estimator of the average treatment effect on the treated is then
\begin{equation*}
    \widehat{\ATT}_{ra} = \widehat m_1-\widehat\tau_{ra}.
\end{equation*}
In \Cref{ass:ra-regularity} in the Supplementary Appendix, we provide additional regularity conditions.  The following proposition provides the asymptotic properties of $\widehat{\ATT}_{ra}$.

\begin{proposition}\label{prop:imputation-an} Under Assumptions \ref{ass:sampling}, \ref{ass:conditional-parallel-trends}, \ref{ass:cov-unc}, \ref{ass:overlap2}, \ref{ass:imputation-linearity}, and \ref{ass:ra-regularity}, the imputation estimator $\widehat{\ATT}_{ra}$ is consistent and satisfies
\begin{align*}
    \sqrt{n}\big(\widehat{\ATT}_{ra} - \ATT\big) \xrightarrow{d} \mathcal{N}(0, \boldsymbol{\Omega}),
\end{align*}
where $\boldsymbol{\Omega} = \mathrm{Var}(\psi^{ra})$ and $\psi^{ra}$ is the influence function defined in Equation \ref{eq:psi-ra-final} in the Supplementary Appendix.
\end{proposition}

\subsection{Neyman Orthogonal Estimator}\label{sec:doubly-robust}\label{sec:estimation-inference:ml}

The imputation estimator in \Cref{sec:estimation-inference:ra} is consistent under correct specification of $m_0(X_{t^*},X_{t^*-1},Z)$ and $\E[X_{t^*} \mid X_{t^*-1}, W, Z, D=0]$. If these models are misspecified, however, the resulting bias does not vanish with sample size. In this section, we provide a Neyman Orthogonal estimator for the $\ATT$, allowing for (i) doubly robust estimation (i.e., consistent estimation when some of the nuisance functions are misspecified), and (ii) using machine learning to estimate the nuisance functions. We introduce the following additional notation for this section.  Let $\pi := \P(D=1)$, and let
\begin{align}
    p(X_{t^*-1},W,Z) = \P(D=1 \mid X_{t^*-1}, W, Z). \label{eqn:p2-def}
\end{align}
$p$ is the propensity score conditional on the covariates $(X_{t^*-1},W,Z)$.  Additionally, define
\begin{align}\label{eqn:omega-def}
    \omega_0\big(X_{t^*}(0),X_{t^*-1},Z\big) := \E\left[\frac{p(X_{t^*-1},W,Z)}{1-p(X_{t^*-1},W,Z)} \midbar X_{t^*}(0), X_{t^*-1}, Z, D=0\right].
\end{align}
$\omega_0$ is the conditional expectation of the odds ratio based on $p(X_{t^*-1}, W, Z)$ given $(X_{t^*}, X_{t^*-1}, Z)$ for the untreated group.  To conserve on notation, let $O = (Y_{t^*-1}, Y_{t^*}, X_{t^*-1}, X_{t^*}, W, Z, D)$ and define
\begin{align}
    \varphi_1(O;\eta) &= \frac{D}{\pi} \Delta Y_{t^*} - \frac{D}{\pi} \nu_0(X_{t^*-1},W,Z) \nonumber \\
    &\quad - \frac{1-D}{\pi} \left(m_0(X_{t^*},X_{t^*-1},Z)-\nu_0(X_{t^*-1},W,Z)\right) \frac{p(X_{t^*-1},W,Z)}{1-p(X_{t^*-1},W,Z)} \nonumber \\
    & \quad - \frac{1-D}{\pi} \left(\Delta Y_{t^*}-m_0(X_{t^*},X_{t^*-1},Z)\right) \omega_0(X_{t^*},X_{t^*-1},Z). \label{eqn:eif1}
\end{align}

The following proposition provides an alternative estimand for the $\ATT$.

\begin{proposition}\label{prop:dr-att-delta} Under \Cref{ass:sampling,ass:conditional-parallel-trends,ass:cov-unc,ass:overlap2},
\begin{align*}
    \ATT = \E\big[ \varphi_1(O;\eta)\big]
\end{align*}
\end{proposition}

Next, let $\widehat\eta = (\widehat{m}_0, \widehat\nu_0, \widehat p, \widehat\omega_0)$ denote estimators of the nuisance functions in \Cref{eqn:eif1}, $\widehat\pi = n^{-1}\sum_{i=1}^n D_i$, and $\widehat\tau = n_1^{-1} \sum_{i=1}^n D_i\,\widehat\nu_{0,i}$,
where $n_1=\sum_{i=1}^n D_i$ and $\widehat\nu_{0,i} = \widehat\nu_0(X_{i,t^*-1},W_i,Z_i)$. The sample analog of the expression in \Cref{prop:dr-att-delta} is
{\small \begin{align}
    \widehat{\ATT}_{dr} &= \widehat m_1 - \widehat\tau
       - \frac{1}{n}\sum_{i=1}^n \Bigg[\frac{1-D_i}{\widehat\pi} \left(\widehat{m}_{0,i}-\widehat\nu_{0,i}\right) \frac{\widehat p_{i}}{1-\widehat p_{i}} + \frac{1-D_i}{\widehat\pi} \left(\Delta Y_{it^*}-\widehat{m}_{0,i}\right) \widehat\omega_{0,i} \Bigg],
    \label{eq:att-dr-sample}
\end{align}
}where $\widehat m_1$ is the treated mean of $\Delta Y_{t^*}$ as in \Cref{eq:ra-m1hat}, and the remaining shorthand is $\widehat{m}_{0,i} = \widehat{m}_0(X_{it^*},X_{i,t^*-1},Z_i)$, $\widehat p_{i} = \widehat p(X_{i,t^*-1},W_i,Z_i)$, and $\widehat\omega_{0,i} = \widehat\omega_0(X_{it^*},X_{i,t^*-1},Z_i)$.

Next, we provide a double robustness result for $\widehat{\ATT}_{dr}$.  Let $\eta(\theta) = \big(m_0(\theta_m), \nu_0(\theta_\nu), p(\theta_p), \omega_0(\theta_\omega)\big)$ denote parametric working models for the nuisance functions.  A leading example would be to use the same linear models as in \Cref{ass:imputation-linearity} (from our imputation estimator) for $m_0(\theta_m)$ and $\nu_0(\theta_\nu)$, $p(\theta_p)$ to be a logit working model for the propensity score, and $\omega_0(\theta_\omega)$ to come from a regression of the estimated odds ratio, $\hat{p}(\hat{\theta}_p)/(1-\hat{p}(\hat{\theta}_p))$ on $X_{t^*},X_{t^*-1}$, and $Z$ using the untreated group.  We allow for the possibility that the working models are misspecified, e.g., $m_0(\theta_m) \neq m_0$.  Let $\hat{\eta}(\hat{\theta}) = \big(\hat{m}_0(\hat{\theta}_m), \hat{\nu}_0(\hat{\theta}_\nu), \hat{p}(\hat{\theta}_p), \hat{\omega}_0(\hat{\theta}_\omega)\big)$ denote the estimated nuisance functions given the parametric working models discussed above, and let $\widehat{\ATT}_{dr}(\hat{\theta})$ denote the estimator based on using the parametric working models in \Cref{eq:att-dr-sample}.  Under mild regularity conditions, we have that $\hat{\eta}(\hat{\theta}) \xrightarrow{p} \eta(\theta)$.  We also have that $\widehat{\ATT}_{dr}(\hat{\theta}) \xrightarrow{p} \ATT_{dr}(\theta)$, where $\ATT_{dr}(\theta)$ is the probability limit of $\widehat{\ATT}_{dr}(\hat{\theta})$ given the parametric working models for the nuisance functions and allowing for the possibility $\ATT_{dr}(\theta) \neq \ATT$.

\begin{proposition}\label{prop:doubly-robust} Under \Cref{ass:sampling,ass:conditional-parallel-trends,ass:cov-unc,ass:overlap2}, given parametric working models for the nuisance functions $\eta(\theta)$, and assuming that $\hat{\eta}(\hat{\theta}) \xrightarrow{p} \eta(\theta)$ and $\widehat{\ATT}_{dr}(\hat{\theta}) \xrightarrow{p} \ATT_{dr}(\theta)$, and that either
    \begin{itemize}
        \item [(i)] $\big(m_0(\theta_m), \nu_0(\theta_\nu)\big) = \big(m_0, \nu_0\big)$
        \item [(ii)] $\big(p(\theta_p), \omega_0(\theta_\omega)\big) = \big(p, \omega_0\big)$
    \end{itemize}
    Then, $\widehat{\ATT}_{dr}(\hat{\theta}) \xrightarrow{p} \ATT$.
\end{proposition}

\Cref{prop:dr-att-delta} shows that $\widehat{\ATT}_{dr}$ is doubly robust.  Notice that condition (i) in the proposition corresponds to \Cref{ass:imputation-linearity} for our imputation estimator in the previous section.  However, condition (ii) in the proposition provides an alternative setting, effectively based on correctly re-weighting, for $\widehat{\ATT}_{dr}$ to be consistent.  Also, notice that $\nu_0$ and $m_0$ are functionally related.  In general, if $m_0(\theta_m)$ is misspecified, then it will also be the case that $\nu_0(\theta_\nu)$ is misspecified.  Likewise for the re-weighting terms, if $p(\theta_p)$ is misspecified, then $\omega_0(\theta_\omega)$ will generally be misspecified as well.

Next, we move to estimating $\ATT$ by estimating the nuisance functions $\eta$ using machine learning, based on the estimand in \Cref{prop:dr-att-delta}.  This approach allows us to bypass any strong functional form assumptions in estimating $\ATT$.  We summarize our estimator in the following algorithm.

\begin{center}
\fbox{\begin{minipage}{0.92\textwidth}
\small
\textbf{Algorithm 1: DDML Estimator of $\ATT$ under Covariate Unconfoundedness}

\medskip
\textbf{Input:} Data $\{(\Delta Y_{it^*}, D_i, X_{it^*}, X_{i,t^*-1}, W_i, Z_i)\}_{i=1}^n$; number of folds $K$.

\medskip
\begin{enumerate}
\item Partition $\{1,\ldots,n\}$ into $K$ folds $\{\mathcal{I}_1,\ldots,\mathcal{I}_K\}$. Set $\widehat{\pi} = n^{-1}\sum_{i=1}^n D_i$.

\item For each fold $k = 1,\ldots,K$, estimate nuisance functions on training sample $\mathcal{I}_{-k}$:

(a) \textit{First stage.} Estimate $\widehat{m}_0^{-k}$ using untreated units, and $\widehat{p}^{-k}$ using all units.

(b) \emph{Second stage.} Given the first stage estimate of $\widehat{m}_0^{-k}$, estimate $\widehat\nu_0^{-k}$.  Likewise, given the first stage estimate of $\widehat{p}^{-k}$, estimate $\widehat\omega_0^{-k}$, both using untreated units.

(c) \emph{Third Stage.} Evaluate the doubly robust score $\widehat\varphi_{1,i}^{k}$ on held-out fold $\mathcal{I}_k$ based on \Cref{eqn:eif1} with cross-fitted nuisance estimates, $(\widehat{m}_0^{-k}, \widehat\nu_0^{-k}, \widehat p^{-k}, \widehat\omega_0^{-k})$.

\item Compute $\widehat{\ATT}_{dr} = n^{-1}\sum_{k=1}^K \sum_{i \in \mathcal{I}_k} \widehat\varphi_{1,i}^{k}$.

\item Compute $\widehat{V}_{dr} = n^{-1}\sum_{i=1}^n \widehat\varphi_i^2$, and $\text{se}(\widehat{\ATT}) = \sqrt{\widehat{V}_{dr}/n}$, where $\widehat\varphi_i = \widehat\varphi_{1,i} - \widehat{\ATT}_{dr} - \frac{\widehat{\ATT}_{dr}}{\widehat\pi}(D_i-\widehat\pi)$
\end{enumerate}
\end{minipage}}
\end{center}

\bigskip

The above algorithm implements a cross-fitting estimation procedure as in \textcite{chernozhukov-etal-2018}.  We also make the following high-level assumption about the convergence rates of the nuisance functions.

\begin{assumption}[Product Rate for Nuisance Estimators]\label{ass:product-rate}
The nuisance estimators $\widehat\eta = (\widehat{m}_0, \widehat\nu_0, \widehat p, \widehat\omega_0)$ satisfy
\begin{align*}
\text{(i)}\quad &\|\widehat{m}_0 - m_0\|_2 \cdot \|\widehat\omega_0 - \omega_0\|_2 = o_p(n^{-1/2}), \\
\text{(ii)}\quad &\|\widehat\nu_0 - \nu_0\|_2 \cdot \|\widehat p - p\|_2 = o_p(n^{-1/2}), \\
\text{(iii)}\quad &\|\widehat{m}_0 - m_0\|_2 \cdot \|\widehat p - p\|_2 = o_p(n^{-1/2}).
\end{align*}
\end{assumption}
\Cref{ass:product-rate} is similar to the multiplicative convergence rate requirements that are commonly found in the double/de-biased machine learning literature \parencite{chernozhukov-etal-2018}.  All three conditions hold when each nuisance function converges at rate $o_p(n^{-1/4})$ in $L^2$ norm, which is achievable by random forests \parencite{wager-athey-2018} and many other nonparametric and machine learners \parencite{chernozhukov-etal-2018}.

Conditions (i)-(iii) of \Cref{ass:product-rate} are not independent conditions.  For example, Condition (ii) concerns the nested nuisance $\widehat\nu_0$. Because $\nu_0 = T_{m_0} m_0$ where $T_{m_0}$ is the conditional-expectation operator in \Cref{eqn:nu0}, the $L^2$ error $\|\widehat\nu_0 - \nu_0\|_2$ contains the error from using $\widehat{m}_0$ in the true operator $T_{m_0}$, which is bounded by $\|\widehat{m}_0 - m_0\|_2$ since conditional expectations are $L^2$-contractions and the second-stage approximation error $\|(\hat{T}_{m_0} - T_{m_0})m_0\|_2$ from estimating the operator itself. Both terms need to converge fast enough for condition~(ii) to be satisfied. This means that poor performance in estimating $m_0$ can spill over into poor performance in estimating $\nu_0$.  Likewise, given the similar functional relationship between $p$ and $\omega_0$, poor performance in estimating $p$ can lead to poor performance in estimating $\omega_0$.

The following proposition provides the asymptotic properties of our ML estimator under the conditions stated above plus the regularity conditions given in \Cref{ass:dr-regularity} in the Supplementary Appendix.

\begin{proposition}\label{prop:dr-an} Under \Cref{ass:sampling,ass:conditional-parallel-trends,ass:cov-unc,ass:overlap2,ass:product-rate}, and \ref{ass:dr-regularity}, $\widehat{\ATT}_{dr}$ is consistent and satisfies
\begin{align*}
    \sqrt{n}\big(\widehat{\ATT}_{dr}-\ATT\big) = \frac{1}{\sqrt{n}}\sum_{i=1}^n \varphi\left(O_i;\eta\right) + o_p(1)
    \xrightarrow{d} \mathcal{N}(0, V_{dr}),
\end{align*}
where \vspace{-20pt}
\begin{align*}
    \varphi(O_i;\eta) := \Big(\varphi_1(O_i,\eta)-\ATT\Big) - \frac{\ATT}{\pi}(D_i-\pi),
\end{align*}
and $V_{dr} = \mathrm{Var}\big(\varphi(O;\eta)\big)$. %
$\widehat{V}_{dr}$, given in Algorithm 1, is consistent for $V_{dr}$.
\end{proposition}

\begin{remark}[Monte Carlo Simulations]
    In Supplementary Appendix \ref{app:montecarlo}, we provide Monte Carlo simulations that evaluate the finite sample performance of both the imputation estimator and double/de-biased machine learning that were proposed in this section.
\end{remark}

\begin{remark}[TWFE Regressions]
    The estimators above build on the modern DiD estimators in \textcite{santanna-zhao-2020,callaway-santanna-2021}, rather than traditional two-way fixed effects (TWFE) regressions such as
    \begin{align*}
        Y_{it} = \theta_t + \eta_i + \alpha D_{it} + X_{it} \beta + Z_i \delta_t + e_{it}
    \end{align*}
    where $\theta_t$ is a time fixed effect, $\eta_i$ is a unit fixed effect, and $\alpha$ is the coefficient of interest.  \textcite{caetano-callaway-2025} study this type of specification in detail and point out limitations due to treatment effect heterogeneity and hidden linearity bias (bias that arises because differencing out the unit fixed effect also differences the time-varying covariates).  In our setting with a bad control, an additional bias term arises due to it being affected by the treatment, similar to our discussion in \Cref{sec:approach-use-bad-control}.  If the researcher does not include the bad control in the TWFE regression, then a different bias term arises, similar to our discussion in \Cref{sec:approach-discard-bad-control}.  Finally, because our approach above effectively dealt with the bad control in a distinct step from recovering the $\ATT$, it seems difficult to operationalize a TWFE version of our approaches \Cref{sec:identification,sec:staggered}.
\end{remark}

\section{Application to Job Displacement} \label{sec:application}

We illustrate the methods developed in the paper in an application to job displacement, using data from the National Longitudinal Survey of Youth 1979 (NLSY79). In particular, we study the effect of job displacement on earnings, building on a large literature in economics \parencite[e.g.,][]{jacobson-lalonde-sullivan-1993,stevens-1997,davis-wachter-2011,couch-placzek-2010,barnette-odongo-reynolds-2021,huckfeldt-2022}. We treat a worker's occupation score---the typical log wage paid by their occupation---as a bad control. Occupation is often thought of as a bad control in the job displacement literature as (i) the parallel trends assumption seems more plausible if it holds conditional on the occupation a worker would hold in the absence of job displacement and (ii) job displacement can shift workers into different occupations \parencite{kambourov-manovskii-2009,gathmann-schoenberg-2010}.

\subsection{Data} \label{sec:application-data}

We use biennial data from 1992--2002.  In 1992, NLSY79 respondents were between 28 and 35 years old.  We use the logarithm of yearly earnings as the outcome.  Following the literature, we define job displacement as a worker involuntarily leaving a job through no fault of their own.  The NLSY tracks respondents' jobs, and when a respondent leaves a job, the NLSY records the reason.  We categorize a worker as being displaced if they left any job in the previous two years where the reason was either ``layoff/job eliminated'' or ``plant/company/office/workplace closed.''  And, for example, this excludes workers who quit, were fired, or whose temporary job ended.  We record, at the biennial level, the first year that a worker was displaced from a job.  By construction, this leads to staggered treatment adoption, which we account for in our results below.

To construct occupation score, we supplement the NLSY data with data from the IPUMS USA 1990 5\% sample.  We construct occupation score as the median log hourly wage among employed wage/salary workers between 16 and 64 in that occupation.  The occupation score is constant for a given occupation across time.  For an individual in our main dataset, we construct their occupation score by mapping their observed occupation into the occupation score that we constructed with the IPUMS USA data.  Since individuals can change occupation across time, the individual's occupation score can change across time.

We use a person's race, sex, and years of education as additional covariates.  Our final sample is a balanced panel of 3,231 individuals with positive earnings and non-missing data on the outcome, bad control, and other covariates in every period.  Summary statistics are provided in \Cref{tab:summary_statistics}.  On average, displaced workers have lower earnings than non-displaced workers in 1992 (before any workers in our sample are displaced) and in 2002, with the gap being somewhat larger in 2002 than in 1992.  Similarly, displaced workers have occupations with a lower occupation score than non-displaced workers in 1992 and 2002, though the difference is roughly constant over time.  Compared to non-displaced workers, displaced workers are also less educated, more likely to be black or Hispanic, and less likely to be female.

\setlength{\tabcolsep}{10pt}
\begin{table}[t]
\centering
\caption{Summary Statistics}
\label{tab:summary_statistics}
\begin{threeparttable}
{\small
\begin{tabular}{lcccc}
\toprule
& Displaced & Non-displaced & Diff & se(Diff) \\
\midrule
Log earnings, 1992 & 9.74 & 9.91 & -0.17 & 0.04 \\
 & (0.88) & (0.79) & & \\
Log earnings, 2002 & 10.30 & 10.52 & -0.22 & 0.03 \\
 & (0.84) & (0.82) & & \\
Occupation score, 1992 & 2.25 & 2.32 & -0.07 & 0.01 \\
 & (0.33) & (0.35) & & \\
Occupation score, 2002 & 2.30 & 2.36 & -0.06 & 0.01 \\
 & (0.32) & (0.34) & & \\
Black (\%) & 29.4 & 23.8 & 5.7 & 1.9 \\
 & (45.6) & (42.6) & & \\
Hispanic (\%) & 19.9 & 16.2 & 3.7 & 1.6 \\
 & (40.0) & (36.9) & & \\
Female (\%) & 43.4 & 47.3 & -3.8 & 2.1 \\
 & (49.6) & (49.9) & & \\
Education (years) & 13.57 & 14.12 & -0.56 & 0.10 \\
 & (2.32) & (2.63) & & \\[5pt]
N & 748 & 2,483 & & \\
\bottomrule
\end{tabular}
}
\begin{tablenotes}
\item \footnotesize \textit{Notes:} The table reports summary statistics for displaced and non-displaced workers.  The column ``Displaced'' refers workers who were displaced from their job in any period from 1994--2002.  ``Non-displaced'' refers to workers who are not displaced in any period from 1992--2002 (we drop workers who become displaced in 1992).  Each column reports the mean and standard deviation (in parentheses) for displaced and non-displaced workers, respectively.  For log earnings and occupation score, the year indicates which year the mean and standard deviation are for.  The column ``Diff'' reports the difference between the means for the two groups, and the column ``se(Diff)'' reports its standard error.
\end{tablenotes}
\end{threeparttable}
\end{table}

\subsubsection*{Implementation Details}

For our main results below, we estimate the effect of job displacement on $\log(Earnings)$ using nine different estimators:

\begin{enumerate}
    \item \texttt{TWFE:\,include BC} --- TWFE regression that includes $D_{it}$ and $X_{it}$ as regressors
    \item \texttt{TWFE:\,exclude BC} --- TWFE regression that only includes $D_{it}$ as a regressor
    \item \texttt{Imp:\,include BC} --- Imputation estimator based on \textcite{callaway-santanna-2021} that includes the bad control in both periods as a covariate.
    \item \texttt{Imp:\,exclude BC} --- Imputation estimator based on \textcite{callaway-santanna-2021} that does not include the bad control in either period as a covariate.
    \item \texttt{PT for X} --- Imputation estimator under parallel trends for the bad control as in \Cref{ass:bad-control-parallel-trends} and based on \Cref{prop:att-under-bad-control-parallel-trends} in the Supplementary Appendix.
    \item \texttt{ML:\,pre-treatment} --- Machine learning estimator under simple covariate unconfoundedness that includes the bad control in the pre-treatment period as a covariate with group-time average treatment effects estimated based on \textcite{santanna-zhao-2020} with nuisance functions estimated by random forest (we use the same implementation as \texttt{ML} below, but without its bad control nuisance steps).
    \item \texttt{Imputation} --- Imputation estimator under covariate unconfoundedness with group-time average treatment effects estimated based on \Cref{sec:estimation-inference:ra}.
    \item \texttt{DR} --- Doubly robust estimator with group-time average treatment effects estimated based on \Cref{sec:estimation-inference:ml} with nuisance functions specified parametrically.
    \item \texttt{ML} --- Machine learning estimator with group-time average treatment effects estimated based on \Cref{sec:estimation-inference:ml} with nuisance functions estimated using machine learning.\footnote{In particular, we use cross-fitting with five folds, estimating each of the nuisance functions by random forests using the R \texttt{grf} package \parencite{grf-2026} with \texttt{grf}'s default tuning parameters. One complication is that $\hat\nu_0^{-k}$ and $\hat\omega_0^{-k}$ are second-step regressions where the outcome depends on the first-step estimates $\hat m_0^{-k}$ and $\hat p^{-k}$, respectively.  This can result in overfitting $\hat\nu_0^{-k}$ and $\hat\omega_0^{-k}$ when the same units are used to estimate the outcome and then again for $\hat\nu_0^{-k}$ or $\hat\omega_0^{-k}$.  To address this, we use each unit's out-of-bag prediction (i.e., its averaged prediction using only first-stage trees that excluded that unit), rather than its predicted value across all trees, as the pseudo-outcome on which $\hat\nu_0^{-k}$ or $\hat\omega_0^{-k}$ is computed. An alternative for machine learners without out-of-bag predictions is to instead further split the training fold itself, as in \textcite{farbmacher-huber-laffers-langen-spindler-2022}.}
\end{enumerate}

The TWFE estimators (1 and 2) do not include any additional covariates, as these are all time-invariant.  The seven other estimators all additionally include the additional covariates (race, sex, and education) that were discussed above.  Estimators 3-9 also all involve estimating group-time average treatment effects and then aggregating them into either the overall ATT or an event study.  Estimators 7-9 are based on our covariate unconfoundedness assumption.  For these estimators, we take $W_i$ to be an individual's $\log(Earnings)$ in the pre-treatment period.\footnote{We implement Estimators 1 and 2 using the R \texttt{fixest} package \parencite{berge-butts-mcdermott-2026}, Estimators 3 and 4 using the R \texttt{ptetools} package \parencite{ptetools-2026}, and Estimators 5-9 using our R \texttt{badcontrols} package \parencite{badcontrols-2026}.}

\subsubsection*{Is the Occupation Score a Bad Control?}

A natural starting point for our analysis is to check whether or not occupation score is affected by the treatment.  To answer this question, we maintain the staggered adoption version of our covariate unconfoundedness assumption in \Cref{ass:multi-cov-unc}.  Then, we implement the first step of our imputation estimator in \Cref{sec:estimation-inference}, which yields an estimate of the overall $\ATT$ of job displacement on occupation score as well as an event study.  These results are provided in \Cref{fig:occ_score_event_study}.

\begin{figure}[th]
    \caption{Occupation Score Event Study}
    \label{fig:occ_score_event_study}
    \begin{center}
        \includegraphics[width=0.75\textwidth]{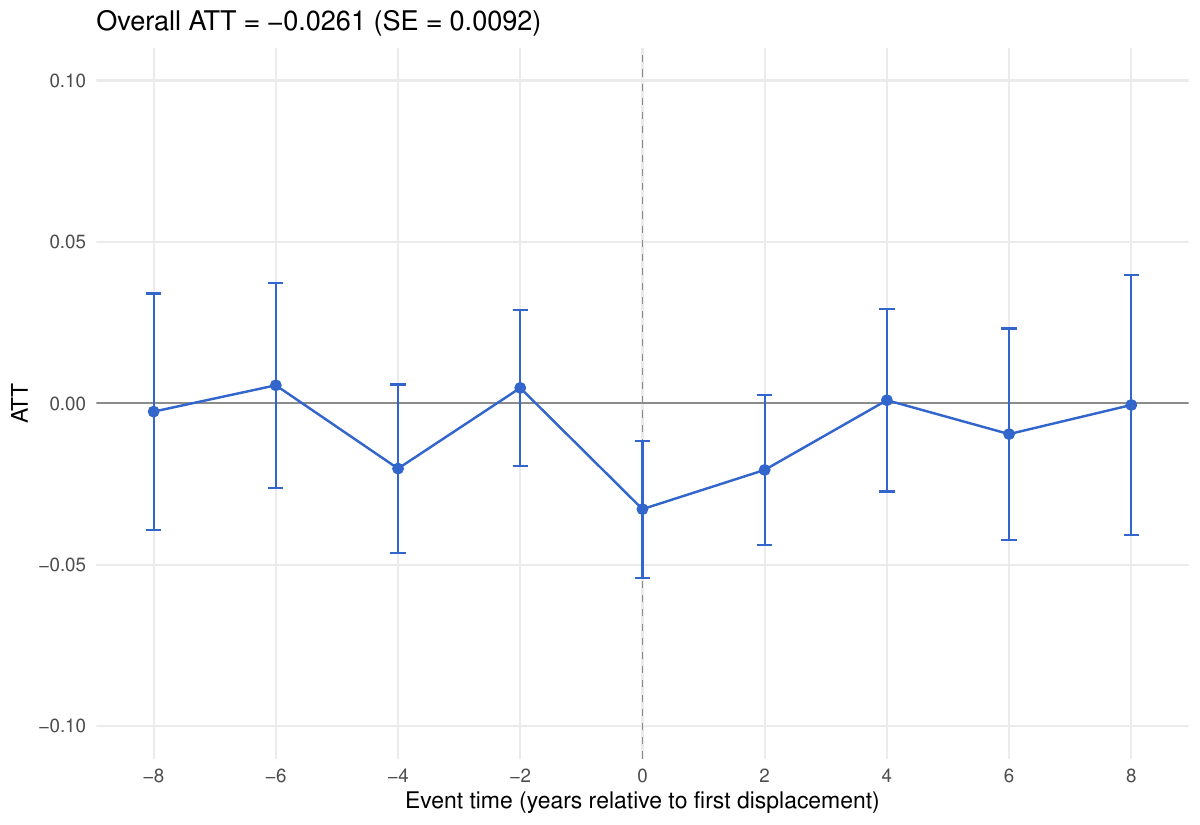}
    \end{center}
    \footnotesize \textit{Notes:} The figure provides estimates of the overall ATT and an event study for occupation score using an imputation estimator based on \Cref{ass:multi-cov-unc}, as discussed in the text.  $e=0$ denotes the period when job displacement occurred.  Post-treatment estimates use the period immediately before displacement as the base period. Pre-treatment estimates are pseudo-$\ATT$'s that instead use the immediately preceding period as the base period (i.e., each is the estimate that would be obtained by treating that period as the start of displacement).  See Remark~\ref{rem:pseudo-att} for more details.
\end{figure}

Our estimates indicate that for displaced workers, on average, job displacement results in their moving to an occupation with 2.6\% lower wages than they would have been in absent job displacement.  Since not all workers change occupations due to job displacement, the effect for workers that do change occupations would be larger.  Turning to the event study, in line with what seems a reasonable expectation, the effect of job displacement is largest in the period in which job displacement occurs, with the effect of job displacement on occupation score reducing in subsequent periods.

\subsubsection*{Main Results}

Next, we provide our estimates of the effect of job displacement on earnings.  We start by providing estimates of the overall $\ATT$ and an event study using our imputation estimator based on covariate unconfoundedness (from Estimator 7 above).

\begin{figure}[th]
    \caption{Imputation Event Study for $\log(Earnings)$}
    \label{fig:imputation_event_study}
    \begin{center}
        \includegraphics[width=0.75\textwidth]{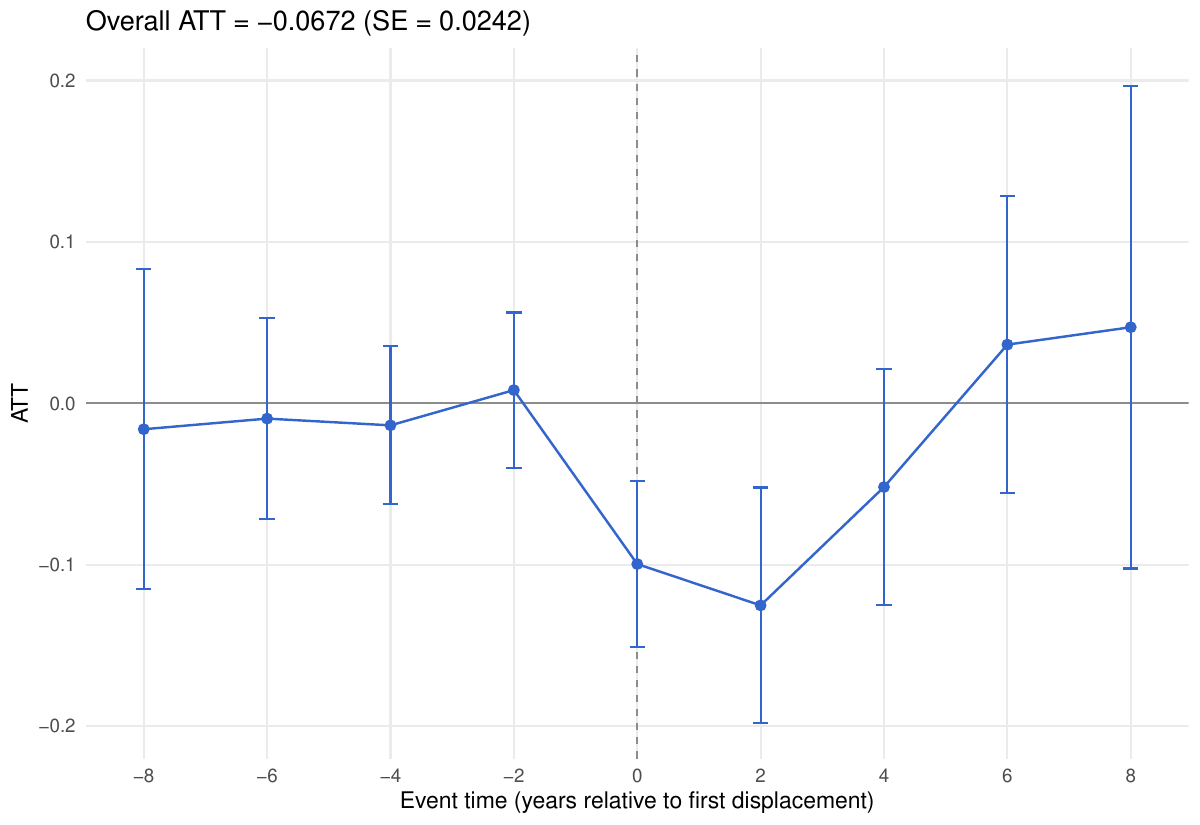}
    \end{center}
    \footnotesize \textit{Notes:} The figure provides estimates of the overall ATT and an event study for $\log(Earnings)$ using our imputation estimator based on covariate unconfoundedness (Estimator 7 above).  $e=0$ denotes the period when job displacement occurred.  Post-treatment estimates use the period immediately before displacement as the base period. Pre-treatment estimates are pseudo-$\ATT$'s that instead use the immediately preceding period as the base period (i.e., each is the estimate that would be obtained by treating that period as the start of displacement).  See Remark~\ref{rem:pseudo-att} for more details.
\end{figure}

The estimates in the figure indicate that, on average, job displacement reduced displaced workers' earnings by about 7\%.  These effects are concentrated in the four years following displacement and are broadly in line with the job displacement literature.

\begin{figure}[th]
    \caption{Overall ATT Estimates Relative to Imputation}
    \label{fig:imputation_diff}
    \begin{center}
        \includegraphics[width=0.75\textwidth]{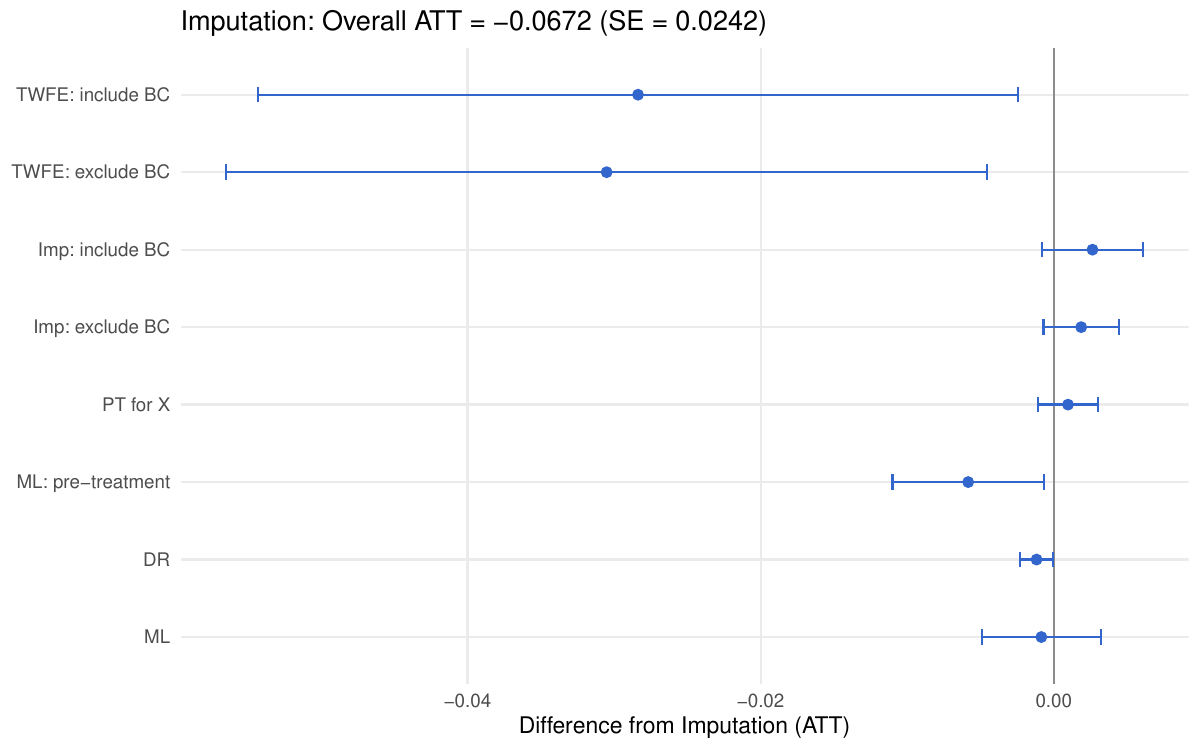}
    \end{center}
    \footnotesize \textit{Notes:} The figure provides the differences between the estimates of the overall $\ATT$ from all of the nine estimators discussed above relative to the imputation estimator in \Cref{fig:imputation_event_study}.  For example, positive differences indicate that the estimator is larger (closer to zero) than our baseline imputation estimate.
\end{figure}

Our main interest, however, is in how sensitive the results are to the different estimators that we discussed above.  We focus on differences in estimates of the overall ATT, and we provide these results in \Cref{fig:imputation_diff}.  To start with, all of the results are qualitatively similar---across estimators, the overall effect of job displacement is always estimated to be between 6 and 10\% lower earnings on average (in unreported event studies, the pattern is similar across estimators too where the effect of job displacement is largest in the first four years following job displacement and smaller and statistically insignificant in subsequent periods).\footnote{Note also that the effects of job displacement are quite large in general which provides a further reason why the qualitative results are similar across estimators.  In an application with smaller main effects, the qualitative results would tend to be more sensitive to the estimator.}  Still, there are reasonably large quantitative differences across estimators.  Both TWFE estimators estimate that job displacement reduces earnings by almost 10\%, resulting in estimates that are roughly 40\% larger in magnitude than the imputation estimates.  The imputation estimators that include or exclude bad control are more similar to our baseline imputation estimator that accounts for occupation being a bad control.  They are both 3-4\% smaller in magnitude than our baseline estimate.  Interestingly, both the TWFE estimates and the imputation estimates that include or exclude the bad control are closer to each other than they are to our baseline imputation estimator, which arguably indicates that the common robustness check of estimating separate regressions that include or exclude the bad control does not necessarily indicate robustness to the tension between including and excluding the bad control that is a main motivation of our paper.

Assuming parallel trends for the bad control (Estimator 5) provides a similar estimate to our baseline imputation estimate.  Under simple covariate unconfoundedness, our ML estimate (Estimator 6) is 9\% larger in magnitude than our baseline imputation estimate, and the difference is statistically significant.  Finally, our doubly robust and ML estimates, which are notably the only ones in \Cref{fig:imputation_diff} that are based on the same covariate unconfoundedness assumption as the baseline imputation estimate, are also quite similar to our baseline imputation estimates.

\section{Conclusion} \label{sec:conclusion}

This paper has considered difference-in-differences identification strategies when the parallel trends assumption holds only after conditioning on time-varying covariates that may have been themselves affected by the treatment. We have shown that the common empirical practices of either including or excluding bad controls can both lead to biased estimates of the average treatment effect on the treated.  We provided explicit conditions under which the approach of conditioning only on the pre-treatment value of the bad control suffices for identification.  We also provided a generalization of that approach based on a version of unconfoundedness for the bad control.

{
\singlespacing
\small
\printbibliography
}

\appendix\crefalias{section}{appendix}
\numberwithin{equation}{section}
\numberwithin{proposition}{section}
\numberwithin{theorem}{section}
\numberwithin{corollary}{section}
\numberwithin{assumption}{section}
\numberwithin{remark}{section}

\section{Identification Proofs} \label{app:proofs}

As a starting point for the results below, notice that
\begin{align}
      \ATT &= \E[Y_{t^*}(1) - Y_{t^*}(0) \mid D=1] \nonumber \\
      &= \E[Y_{t^*}(1) - Y_{t^*-1}(0) \mid D=1] - \E[Y_{t^*}(0) - Y_{t^*-1}(0) \mid D=1] \nonumber \\
      &= \E[\Delta Y_{t^*} \mid D=1] - \E[\Delta Y_{t^*}(0) \mid D=1], \label{eqn:att-in-differences}
\end{align}
which holds by the definition of $\ATT$, adding and subtracting $\E[Y_{t^*-1}(0) \mid D=1]$, and by replacing potential outcomes with their observed counterparts.  \Cref{eqn:att-in-differences} highlights that identification of the $\ATT$ comes down to identifying $\E[\Delta Y_{t^*}(0) \mid D=1]$.

\begin{proof}[\textbf{Proof of \Cref{thm:att-pretreatment-covs}}]

    Notice that
    \begin{align*}
        \E[\Delta Y_{t^*}(0) \mid D=1] &= \E\Big[ \E\big[\Delta Y_{t^*}(0) \mid X_{t^*-1}, Z, D=1 \big] \Bigm| D=1 \Big] \\
        &= \E\Big[ \E\big[ \E[\Delta Y_{t^*}(0) \mid X_{t^*}(0), X_{t^*-1}, Z, D=1 ] \bigm| X_{t^*-1}, Z, D=1 \big] \Bigm| D=1 \Big] \\
        &= \E\Big[ \E\big[ \E[\Delta Y_{t^*}(0) \mid X_{t^*}(0), X_{t^*-1}, Z, D=0 ] \bigm| X_{t^*-1}, Z, D=1 \big] \Bigm| D=1 \Big] \\
        &= \E\Big[ \E\big[ \E[\Delta Y_{t^*}(0) \mid X_{t^*}(0), X_{t^*-1}, Z, D=0] \bigm| X_{t^*-1}, Z, D=0\big]  \Bigm| D=1 \Big] \\
        &= \E\Big[ \E\big[\Delta Y_{t^*}(0) \mid X_{t^*-1}, Z, D=0 \big] \Bigm| D=1 \Big] \\
        &= \E\Big[ m_0(X_{t^*-1}, Z) \Bigm| D=1 \Big]
    \end{align*}
    where the first and second equalities both hold by applying the law of iterated expectations, the third equality holds by \Cref{ass:conditional-parallel-trends}, the fourth equality holds by \Cref{ass:simple-cov-unc}, the fifth equality holds by the law of iterated expectations, and the last equality holds by the definition of $m_0$.  Given this expression for $\E[\Delta Y_{t*}(0) \mid D=1]$, identification of the $\ATT$ follows immediately from \Cref{eqn:att-in-differences}.
\end{proof}

\begin{proof}[\textbf{Proof of \Cref{thm:att-cov-unc}}]
    Notice that
    \begin{align*}
        \E[\Delta Y_{t^*}(0) \mid D=1]
        &= \E\Big[ \E[\Delta Y_{t^*}(0) \mid X_{t^*}(0), X_{t^*-1}, Z, D=1] \Bigm| D=1 \Big] \\
        &= \E\Big[ \E[\Delta Y_{t^*}(0) \mid X_{t^*}(0), X_{t^*-1}, Z, D=0] \Bigm| D=1 \Big] \\
        &= \E\Big[ m_0\big(X_{t^*}(0), X_{t^*-1}, Z\big) \Bigm| D=1 \Big] \\
        &= \E\Big[ \E\big[ m_0\big(X_{t^*}(0), X_{t^*-1}, Z\big) \bigm| X_{t^*-1}, W, Z, D=1 \big] \Bigm| D=1 \Big] \\
        &= \E\Big[ \E\big[\, m_0\big(X_{t^*}(0), X_{t^*-1}, Z\big) \bigm| X_{t^*-1}, W, Z, D=0 \big] \Bigm| D=1 \Big] \\
        &= \E\Big[ \E\big[\, m_0\big(X_{t^*}, X_{t^*-1}, Z\big) \bigm| X_{t^*-1}, W, Z, D=0 \big] \Bigm| D=1 \Big]
    \end{align*}
    where the first equality holds by the law of iterated expectations, the second equality holds by \Cref{ass:conditional-parallel-trends}, the third equality holds by the definition of $m_0$, the fourth equality holds by the law of iterated expectations, the fifth equality holds by \Cref{ass:cov-unc}, and the last equality holds because the untreated potential bad control is observed for the untreated group.
\end{proof}

\begin{lemma} \label{lem:mp1} Under \Cref{ass:staggered,ass:no-anticipation,ass:multi-sampling,ass:multi-conditional-parallel-trends,ass:multi-overlap}, for any post-treatment period $t\geq g$ and any not-yet-treated group $g'>t$,
    \begin{align*}
        ATT(g,t) = \E[Y_t - Y_{g-1} \mid G=g] - \E\left[\E[Y_t - Y_{g-1} \mid \mathbf{X}(0),Z,G=g']\Bigm| G=g\right]
    \end{align*}
\end{lemma}
\begin{proof}
    Notice that
    \begin{align*}
        \ATT(g,t) &= \E[Y_t(g) - Y_t(0) \mid G=g] \nonumber \\
        &= \E[Y_t(g) - Y_{g-1}(0) \mid G=g] - \E[Y_t(0) - Y_{g-1}(0) \mid G=g] \\
        &= \E[Y_t - Y_{g-1} \mid G=g] - \sum_{s=g}^t \E[Y_s(0) - Y_{s-1}(0) \mid G=g] \\
        &= \E[Y_t - Y_{g-1} \mid G=g] - \sum_{s=g}^t \E\Big[\E[Y_s(0) - Y_{s-1}(0) \mid \mathbf{X}(0),Z,G=g] \Bigm| G=g\Big] \\
        &= \E[Y_t - Y_{g-1} \mid G=g] - \E\left[\sum_{s=g}^t\E[Y_s(0) - Y_{s-1}(0) \mid \mathbf{X}(0),Z,G=g']\Bigm| G=g\right] \\
        &= \E[Y_t - Y_{g-1} \mid G=g] - \E\left[\E[Y_t - Y_{g-1} \mid \mathbf{X}(0),Z,G=g']\Bigm| G=g\right]
    \end{align*}
    where the first equality holds by the definition of $\ATT(g,t)$, the second equality holds by adding and subtracting $\E[Y_{g-1}(0) \mid G=g]$ (the outcome for group $g$ in period $g-1$, which is group $g$'s most recent pre-treatment period), the third equality holds because $Y_t(g)$ and $Y_{g-1}(0)$ are observed outcomes for group $g$ and by adding and subtracting $\E[Y_s(0) \mid G=g]$ for $s=g, \ldots t-1$, the fourth equality holds by the law of iterated expectations, the fifth equality changes the order of the summation and expectation and holds by \Cref{ass:multi-conditional-parallel-trends}, and the last equality holds by canceling all the duplicate conditional expectations from $s=g, \ldots, t-1$.
\end{proof}

\begin{lemma} \label{lem:mp-cov-unc2} Under \Cref{ass:multi-cov-unc,ass:multi-overlap},
\begin{align*}
    \big(X_{g}(0), \ldots X_\T(0)\big) \independent G \bigm| X_{g-1}(0), \ldots X_{1}(0), W, Z
\end{align*}
\end{lemma}

\begin{proof}
    The proof follows a very similar argument as for \Cref{lem:mp-cov-unc1} below and is therefore omitted.
\end{proof}

\begin{proof}[\textbf{Proof of \Cref{thm:mp}}]
    Starting from the second term in \Cref{lem:mp1} and taking $g' = \infty$, notice that
    {\small
        \begin{align}
            & \E\left[\E[Y_t - Y_{g-1} \mid \mathbf{X}_{g:\T}(0),\mathbf{X}_{1:(g-1)},Z,G=\infty]\Bigm| G=g\right] \nonumber \\
            & \hspace{0pt} = \E\left[\E\big[\E[Y_t - Y_{g-1} \mid \mathbf{X}_{g:\T},\mathbf{X}_{1:(g-1)},Z,G=\infty]\bigm| \mathbf{X}_{1:(g-1)},W,Z,G=g \big] \Bigm| G=g\right] \nonumber \\
            & \hspace{0pt} = \E\left[\E\big[\E[Y_t - Y_{g-1} \mid \mathbf{X}_{g:\T},\mathbf{X}_{1:(g-1)},Z,G=\infty]\bigm| \mathbf{X}_{1:(g-1)},W,Z,G=\infty \big] \Bigm| G=g\right] \label{eqn:mp-proof3}
        \end{align}
    }where the first equality holds by the law of iterated expectations and because $\mathbf{X}_{g:\T} = \mathbf{X}_{g:\T}(0)$ for the never-treated group, and the second equality holds under \Cref{ass:multi-cov-unc} by \Cref{lem:mp-cov-unc2}.  Plugging this expression into \Cref{lem:mp1} implies the result.
\end{proof}

\begin{lemma} \label{lem:mp-cov-unc1} Under \Cref{ass:multi-cov-unc,ass:multi-overlap},
\begin{align*}
    \big(X_{g}(0), \ldots X_\T(0)\big) \independent G \bigm| X_{g-1}(0), W, Z
\end{align*}
\end{lemma}

\begin{proof}
    Notice that, for any $\mathbf{x} \in \text{support}\big(\mathbf{X}_{g:\T}(0)\big)$,
    \begin{align*}
        & \P\big(\mathbf{X}_{g:\T}(0) \leq \mathbf{x} \mid X_{g-1}(0), W, Z, G\big) \\
        & \hspace{25pt}= \int \indicator{\tilde{\mathbf{x}} \leq \mathbf{x}} \, d\F_{X_g(0),\ldots,X_\T(0) \mid X_{g-1}(0), W, Z, G}(\tilde{x}_g, \ldots, \tilde{x}_\T \mid \cdot) \\
        & \hspace{25pt}= \int \indicator{\tilde{\mathbf{x}} \leq \mathbf{x}} \, d\F_{X_\T(0) \mid X_{\T-1}(0), \ldots, X_{g-1}(0), W, Z, G}(\tilde{x}_\T \mid \cdot) \cdots d\F_{X_g(0) \mid X_{g-1}(0), W, Z, G}(\tilde{x}_g \mid \cdot) \\
        & \hspace{25pt}= \int \indicator{\tilde{\mathbf{x}} \leq \mathbf{x}} \, d\F_{X_\T(0) \mid X_{\T-1}(0), \ldots, X_{g-1}(0), W, Z}(\tilde{x}_\T \mid \cdot) \cdots d\F_{X_g(0) \mid X_{g-1}(0), W, Z}(\tilde{x}_g \mid \cdot) \\
        & \hspace{25pt} = \P\big(\mathbf{X}_{g:\T}(0) \leq \mathbf{x} \mid X_{g-1}(0), W, Z\big)
    \end{align*}
    where the first equality holds by writing the probability as an integral and where $\tilde{\mathbf{x}} := (\tilde{x}_g, \ldots, \tilde{x}_\T)'$, the second equality holds by repeatedly applying the definition of conditional probability, the third equality holds by \Cref{ass:multi-cov-unc}, and the last equality holds by combining terms.  This implies the result.
\end{proof}

\begin{proof}[\textbf{Proof of \Cref{prop:mp-lower-dimension}}]
    Using a similar argument as in the proof of \Cref{lem:mp1} and by invoking \Cref{ass:mp-dim-reduction}, one can show that
    { \small
    \begin{align}
        \ATT(g,t) &= \E[Y_t - Y_{g-1} \mid G=g] - \E\left[\E[Y_t - Y_{g-1} \mid X_t(0), X_{g-1},Z,D_t=0]\Bigm| G=g\right] \label{eqn:mp-prop1}
    \end{align}
    }where the key remaining challenge is that the distribution of $X_t(0) \mid X_{g-1}, Z, G=g$ is not directly identified since $t$ is a post-treatment period for group $g$.  However, for the second term in \Cref{eqn:mp-prop1}, it follows that
    \begin{align*}
        & \E\left[\E[Y_t - Y_{g-1} \mid X_t(0), X_{g-1},Z,D_t=0]\Bigm| G=g\right] \\
        & \hspace{25pt} = \E\left[ \E\big[ \E[Y_t - Y_{g-1} \mid X_t(0), X_{g-1},Z,D_t=0] \bigm| X_{g-1}, W, Z, G=g \big] \Bigm| G=g\right] \\
        & \hspace{25pt} = \E\left[ \E\big[ \E[Y_t - Y_{g-1} \mid X_t(0), X_{g-1},Z,D_t=0] \bigm| X_{g-1}, W, Z, D_t=0 \big] \Bigm| G=g\right]
    \end{align*}
    where the first equality holds by the law of iterated expectations and the second equality holds by \Cref{lem:mp-cov-unc1}.
\end{proof}

\begin{lemma} \label{lem:mp-simple-cov-unc} Under \Cref{ass:multi-overlap,ass:multi-simple-cov-unc},
\begin{align*}
    \big(X_{g}(0), \ldots X_\T(0)\big) \independent G \bigm| X_{g-1}(0), Z
\end{align*}
\end{lemma}

\begin{proof}
    The proof is very similar to that of \Cref{lem:mp-cov-unc1} and is therefore omitted.
\end{proof}

\begin{proof}[\textbf{Proof of \Cref{prop:mp-simple-cov-unc}}]
    To start with, since \Cref{lem:mp1} holds for all not yet treated groups, we have that
    \begin{align}
        ATT(g,t) = \E[Y_t - Y_{g-1} \mid G=g] - \E\left[\E[Y_t - Y_{g-1} \mid \mathbf{X}(0),Z,D_t=0]\Bigm| G=g\right] \label{eqn:mp1-res}
    \end{align}
    First, if \Cref{ass:mp-bad-control-redundancy} holds, then the second term in \Cref{eqn:mp1-res} is given by
    \begin{align*}
        \E\left[\E[Y_t - Y_{g-1} \mid \mathbf{X}(0),Z,D_t=0]\Bigm| G=g\right] = \E\left[\E[Y_t - Y_{g-1} \mid X_{g-1},Z,D_t=0]\Bigm| G=g\right]
    \end{align*}
    which implies the result in this case.

    If, instead, \Cref{ass:multi-simple-cov-unc} holds, for the second term in \Cref{eqn:mp1-res}, it follows that
    \begin{align*}
        & \E\left[\E[Y_t - Y_{g-1} \mid \mathbf{X}(0),Z,D_t=0]\Bigm| G=g\right] \\
        & \hspace{50pt} = \E\left[\E[Y_t - Y_{g-1} \mid X_t(0),X_{g-1},Z,D_t=0]\Bigm| G=g\right] \\
        & \hspace{50pt} = \E\left[ \E\big[ \E[Y_t - Y_{g-1} \mid X_t(0), X_{g-1}, Z,D_t=0] \bigm| X_{g-1}, Z, G=g\big] \Bigm| G=g\right] \\
        & \hspace{50pt} = \E\left[ \E\big[ \E[Y_t - Y_{g-1} \mid X_t(0), X_{g-1}, Z,D_t=0] \bigm| X_{g-1}, Z, D_t=0\big] \Bigm| G=g\right] \\
        & \hspace{50pt} = \E\left[\E\big[Y_t - Y_{g-1} \bigm|  X_{g-1},Z,D_t=0 \big] \Bigm| G=g\right]
    \end{align*}
    where the first equality holds by \Cref{ass:mp-dim-reduction}, the second equality holds by the law of iterated expectations, the third equality holds by \Cref{lem:mp-simple-cov-unc}, and the last equality holds by the law of iterated expectations.  Plugging this expression into \Cref{eqn:mp1-res} implies the result.
\end{proof}

\end{document}